\pdfoutput=1 
\documentclass[acmtog,nonacm,screen,dvipsnames,prologue]{acmart}
\usepackage[utf8]{inputenc}

\copyrightyear{2021}
\setcopyright{none}
\acmYear{}
\acmConference{}{}{}
\acmPrice{}
\acmDOI{}
\acmISBN{}

\usepackage{physics}
\usepackage{amsthm}
\usepackage{mathtools}  
\usepackage{csquotes}

\newcommand{\B}{B}
\newcommand{\W}{W}

\newcommand{\Tbl}{\mathbf{T}}

\newcommand{\bl}{b}
\newcommand{\kbl}{k}

\DeclareMathOperator{\arctantwo}{arctan2}
\DeclareMathOperator{\sinhc}{sinch}

\DeclareMathOperator{\sign}{sign}
\DeclareMathOperator{\si}{\mathbf s}
\DeclareMathOperator{\co}{\mathbf c}
\DeclareMathOperator{\ta}{\mathbf t}

\newcommand{\colora}[1]{\textcolor{black}{x^{#1}}}
\newcommand{\colorb}[1]{\textcolor{RedOrange}{x^{#1}}}
\newcommand{\colorc}[1]{\textcolor{MidnightBlue}{x^{#1}}}
\newcommand{\colord}[1]{\textcolor{Fuchsia}{x^{#1}}}
\newcommand{\colore}[1]{\textcolor{black}{x^{#1}}}

\newcommand{\cev}[1]{\reflectbox{\ensuremath{\vec{\reflectbox{\ensuremath{#1}}}}}}
\newcommand{\GR}[2][{}]{\ensuremath{\mathbb{R}^{{#1}}_{#2}}}
\newcommand{\e}{\mathbf {e}_}            
\newcommand{\Refl}{\mathbf {R}}            

\newcommand{\Spin}[1]{\ensuremath{\text{Spin}({#1})}}
\newcommand{\Pin}[1]{\ensuremath{\text{Pin}({#1})}}
\newcommand{\SU}[1]{\ensuremath{\text{SU}({#1})}}
\newcommand{\Eu}[1]{\ensuremath{\text{E}({#1})}}
\newcommand{\SE}[1]{\ensuremath{\text{SE}({#1})}}
\newcommand{\Conformal}[1]{\ensuremath{\text{Conf}({#1})}}
\newcommand{\SConformal}[1]{\ensuremath{\text{SConf}({#1})}}

\newcommand{\Or}[1]{\ensuremath{\text{O}({#1})}}
\newcommand{\SO}[1]{\ensuremath{\text{SO}({#1})}}
\newcommand{\spin}[1]{\ensuremath{\mathfrak{spin}({#1})}}

\DeclareMathOperator\Ln{Ln}

\DeclarePairedDelimiter\floor{\lfloor}{\rfloor}
\DeclarePairedDelimiter\ceil{\lceil}{\rceil}
\DeclareMathOperator{\arccosh}{arccosh}

\usepackage{pdflscape} 
\usepackage{subcaption}
\usepackage[font=small,labelfont=bf]{caption}
\usepackage{colortbl}
\usepackage{accents}
\usepackage{eqnarray}
\usepackage{cleveref}
\usepackage{epigraph}
\usepackage{awesomebox}
\usepackage{cancel}

\usepackage[h]{esvect}

\usepackage{pdfpages}

\usepackage{calc}  
\usepackage{enumitem}

\crefname{algocf}{algorithm}{algorithms}
\Crefname{algocf}{Algorithm}{Algorithms}
\crefname{figure}{fig.}{figs.}
\Crefname{figure}{Figure}{Figures}

\newtheorem{theorem}{Theorem}

\theoremstyle{definition}
\newtheorem{corollary}{Corollary}[theorem]

\theoremstyle{definition}
\newtheorem{example}{Example}[section]

\begin{document}

\title{Graded Symmetry Groups: Plane and Simple}

\author{Martin Roelfs}
\authornote{Both authors contributed equally to the paper}
\orcid{0000-0002-8646-7693}
\affiliation{%
  \department{Department of Physics}
  \institution{KU Leuven}
  \city{Kortrijk}
  \country{Belgium}
}
\email{martin.roelfs@kuleuven.be}
\author{Steven De Keninck}
\authornotemark[1]
\affiliation{%
  \department{Informatics Institute}
  \institution{University of Amsterdam}
  \city{Amsterdam}
  \country{The Netherlands}
}
\email{steven@enki.ws}

\begin{abstract}
    The symmetries described by Pin groups are the result of combining a finite number of discrete reflections in (hyper)planes.
    The current work shows how an analysis using geometric algebra provides a picture complementary to that of the classic matrix Lie algebra approach, while retaining information about the number of reflections in a given transformation. 
    This imposes a graded structure on Lie groups, not evident in their matrix representation.
    By embracing this graded structure, the invariant decomposition theorem was proven: any composition of $k$ linearly independent reflections can be decomposed into $\ceil{k/2}$ commuting factors, each of which is the product of at most two reflections.
    This generalizes a conjecture by M. Riesz, and has e.g. the Mozzi-Chasles' theorem as its 3D Euclidean special case.
    To demonstrate its utility, we briefly discuss various examples such as Lorentz transformations, Wigner rotations, and screw transformations.
    The invariant decomposition also directly leads to closed form formulas for the exponential and logarithmic function for all Spin groups, and identifies element of geometry such as planes, lines, points, as the invariants of $k$-reflections.
    We conclude by presenting novel matrix/vector representations for geometric algebras \GR{pqr}, and use this in \Eu{3} to illustrate the relationship with the classic covariant, contravariant and adjoint representations for the transformation of points, planes and lines.
\end{abstract}

\maketitle
\renewcommand{\shortauthors}{Roelfs and De Keninck}

\let\svthefootnote\thefootnote
\let\thefootnote\relax\footnotetext{
Authors’ addresses: Martin Roelfs, Department of Physics, KU Leuven, Kortrijk,
Belgium, martin.roelfs@kuleuven.be; Steven De Keninck, Informatics Institute,
University of Amsterdam, Amsterdam, The Netherlands, steven@enki.ws.
}
\let\thefootnote\svthefootnote

\section{Introduction}

Central to this paper is the generalisation of a conjecture by M. Riesz \cite{Riesz1993}, stating that a bivector of an n-dimensional geometric algebra
\GR{pq} can always be decomposed into at most $\floor{\frac n 2}$ simple commuting orthogonal bivectors. 
We extend this conjecture to the wider class of algebras \GR{pqr}, which includes $r$ null basis vectors, and consider the group of all reflections therein, \Pin{p,q,r}.
The resulting theorem then states:
\begin{theorem}[invariant decomposition]\label{th:invariantdecomposition}
A product of $k$ reflections $U = u_1 u_2 \cdots u_k$ can be decomposed into exactly $\ceil{\frac k 2}$ commuting factors. These are $\floor{\frac k 2}$ products of two reflections, and, for odd $k$, one extra reflection. These factors are called \emph{simple}.
\end{theorem}
\noindent
In the three dimensional Euclidean group this says that every $4$-reflection can be decomposed into two commuting 2-reflections, better known as the Mozzi-Chasles' theorem:

\begin{displayquote}``Every three dimensional rigid body motion can be decomposed as a translation along a line followed or preceded by a rotation around the same line.''
\end{displayquote}
Our generalisation shows that this is - in contrast to popular belief - not the simplest example, and a similar statement can be made for
the $3$-reflections in the two dimensional Euclidean group.

\begin{displayquote}
``Every two dimensional $3$-reflection can be decomposed as a translation along a line followed or preceded by a reflection in the same line.''
\end{displayquote}
Key to understanding \cref{th:invariantdecomposition} and its proof is a graded and very geometric perspective on symmetry groups, which we take some time to explain in \cref{E2}. 
Following Hamilton, we build all isometric transformations by composing reflections, while devoting extra attention to the graded structure this imposes. 
The factors of any isometry are reflections and \emph{bireflections}: a pair of reflections that form either a rotation, translation, or hyperbolic rotation (boost). 
We will not only prove \cref{th:invariantdecomposition}, but also provide an analytical solution for the decomposition of $k$ reflections into $\ceil{\frac k 2}$ commuting simple orthogonal factors.
Having such a geometrically inspired decomposition makes many algebraic operations, such as computing exponentials and logarithms, much easier. 

Paul Dirac famously remarked that his research work was done in pictures, and that he often thought projective geometry the most useful, but 
\begin{displayquote}
    ``When I came to publish the results I suppressed the projective geometry as the results could be expressed more concisely in analytic form.'' -- P.~A.~M.~Dirac~\cite{Dirac} 
\end{displayquote}
But the pictures provide additional geometric insights which are not easily gained purely from algebra.
For example, while the Clifford-Lipschitz group and the twisted Clifford-Lipschitz group might be algebraically isomorphic \cite{Vaz:2016qyw}, the geometrical interpretation makes it clear that the conjugation law of the twisted Clifford-Lipschitz group has to be used to apply reflections, see \cref{sec:conjugation}.
The emphasis on geometry is motivated further by recent advances in computer graphics \cite{10.1145/3305366.3328099,GunnThesis}, which demonstrate that vectors can be identified with (hyper)planes instead of points, an idea dating back to Michel Chasles \cite{chasles1875}, but not often considered.
A pictorial approach will help to underscore the importance of this insight.

In \cref{sec:geometric_algebra} we introduce Clifford algebras, whose graded structure and concise expression for reflections makes them the ideal algebraic framework for formalizing graded symmetry groups. In \cref{eigen} we establish blades as the natural choice to represent the primitive elements of geometry such as points, lines, spheres, etc. We then turn our attention to Spin Lie algebras, and their identification as bivector algebras in \cref{liealgebra}. 
A proof and novel algorithm for the \emph{invariant decomposition} is presented in \cref{sec:invariant_decomposition}.
In \cref{sec:logarithm,sec:exponential,sec:factorization} we use this invariant decomposition to present closed form solutions for exponentials, factorization of group elements, and logarithms, respectively.
We conclude in \cref{matrix} with a novel method to construct efficient matrix-matrix and matrix-vector representations of graded symmetry groups, and show how these contain the classic covariant, contravariant and adjoint matrix representations.
All the analytical results of this work were also condensed into a cheat sheet, which can be found on the last page of this manuscript.

\clearpage

\section{Geometric Intuition}\label{E2}

\subsection{The Euclidean Group \Eu{2}}

To build an intuition that will carry over from the Orthogonal group all the way to the Conformal group, we first study the distance preserving transformations of the plane. We build on Hamilton's observation that rotations and translations can be constructed by composing reflections, and use that idea as our guiding principle. \Cref{figure_refl} illustrates how the reflection of a shape in a line (left) is the basic building block from which both the rotations (middle) and translations (right) can be constructed through composition. The resulting \emph{bireflection} will transform the shape with twice the angle or distance between the lines respectively.      
\begin{figure}[t]
  \includegraphics[width=1.0\columnwidth]{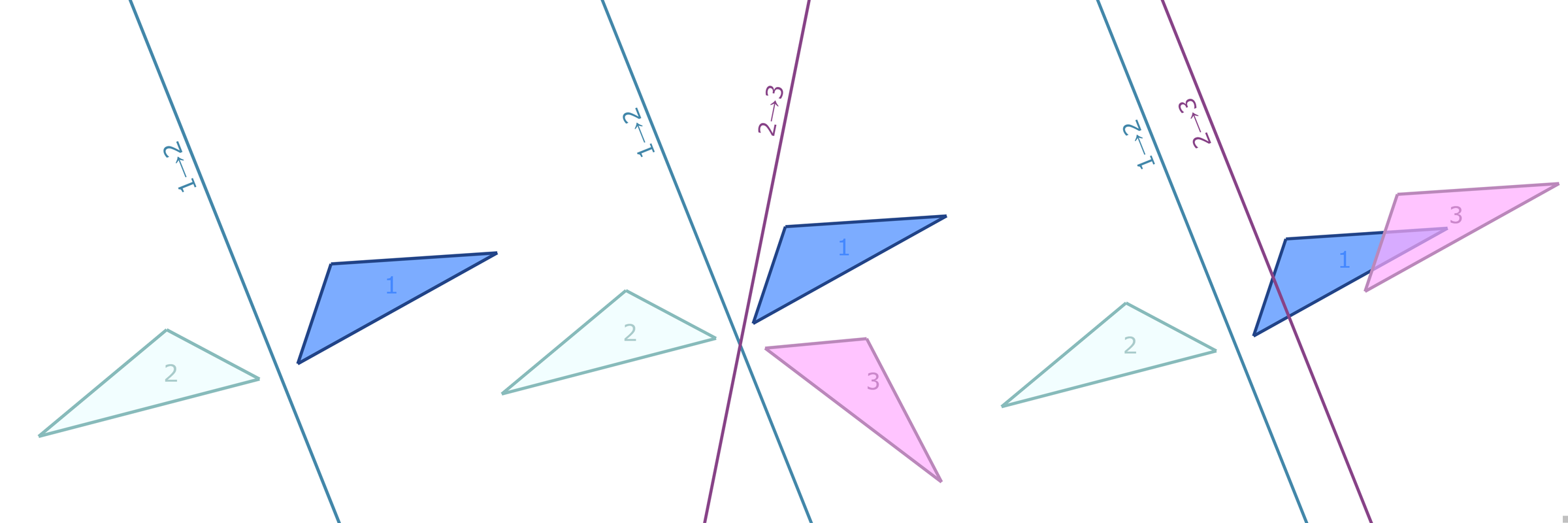}
  \captionof{figure}{Reflections in 2D. From left to right, a single line reflection, two intersecting line reflections forming a rotation, and two parallel line reflections forming a translation.}
  \label{figure_refl}
\end{figure}

\Cref{figure_freedom} illustrates how, when creating a bireflection, only the intersection point and the relative separation between the reflections matters. Indeed, the same rotation can be created using any of the configurations in \cref{figure_freedom}.
This is important, as this \emph{gauge} degree of freedom allows us to select a favorable factorization of any bireflection.
Because of the associativity of reflections, this means that $k$ reflections have $k-1$ gauge degrees of freedom.
\begin{figure}[h]
  \noindent
  \includegraphics[width=1.0\columnwidth]{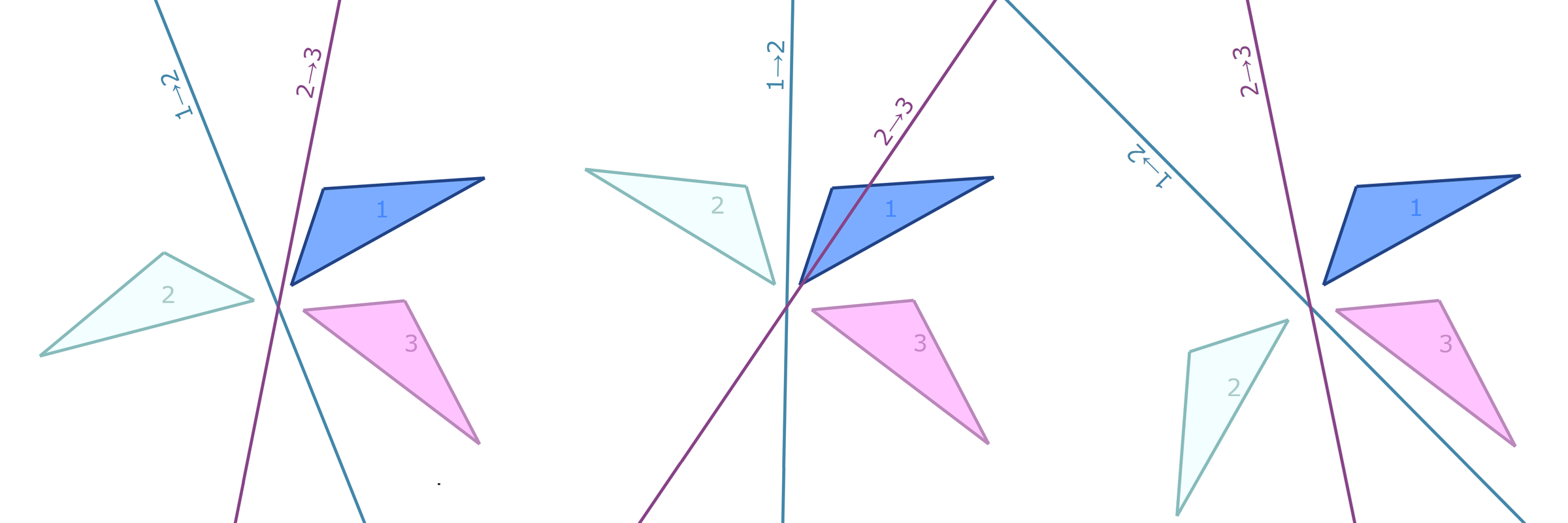}
  \centering
  \caption{Three different sets of line reflections that intersect in the same point and at the same angle create the same rotation around that point.}
  \label{figure_freedom}
\end{figure}
We also note at this point how this approach unifies the treatment of rotations and translations, showing clearly not only how small changes in one of the elementary reflections creates the continuous rotational and translational symmetries, but that indeed rotations and translations are part of the same continuous manifold. 
This relation is illustrated further in \cref{figure_infinity}, which shows how translations can be understood as rotations around infinite or vanishing points. 
\begin{figure}[ht]
  \noindent
  \centering
  \includegraphics[width=1.0\columnwidth]{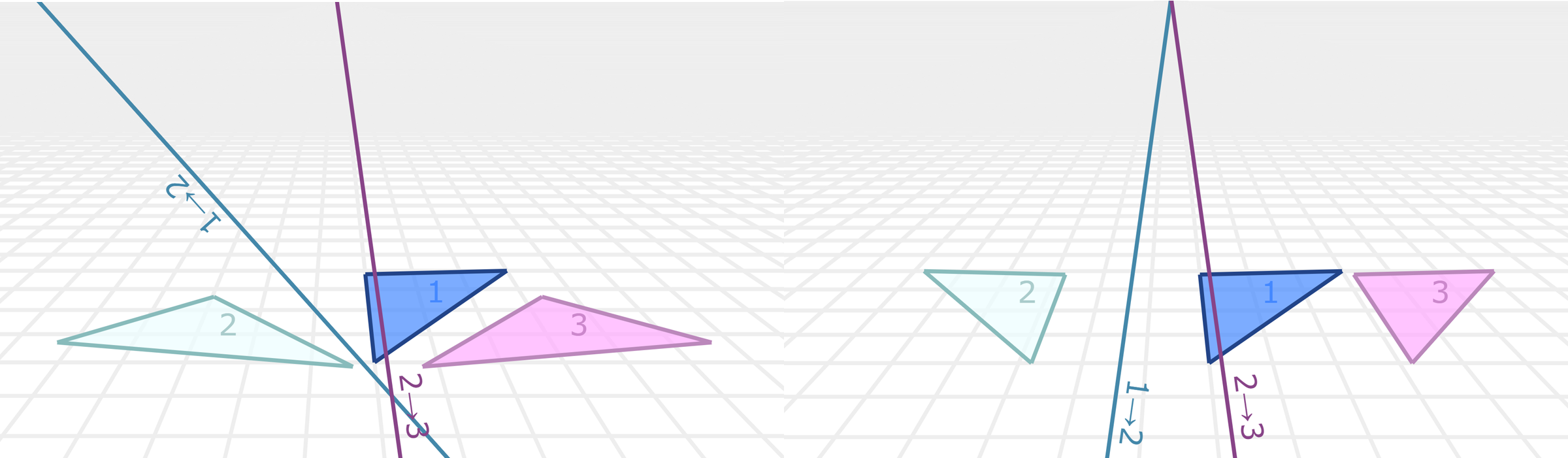}
  \caption{When viewed in 3D, translations (right) can be seen as rotations (left) around infinite points.}
  \label{figure_infinity}
\end{figure}

\subsection{Cartan–Dieudonné}\label{Cartan–Dieudonné}
Having constructed our translations and rotations as bireflections, we ask which isometries can be created by composing a larger number of reflections. For our 2D example, we find in \cref{figure_isometries} how in the plane, the composition of three reflections still produces a new isometry, the glide reflection.
\begin{figure}[h]
  \noindent
  \centering
  \includegraphics[width=1.0\columnwidth]{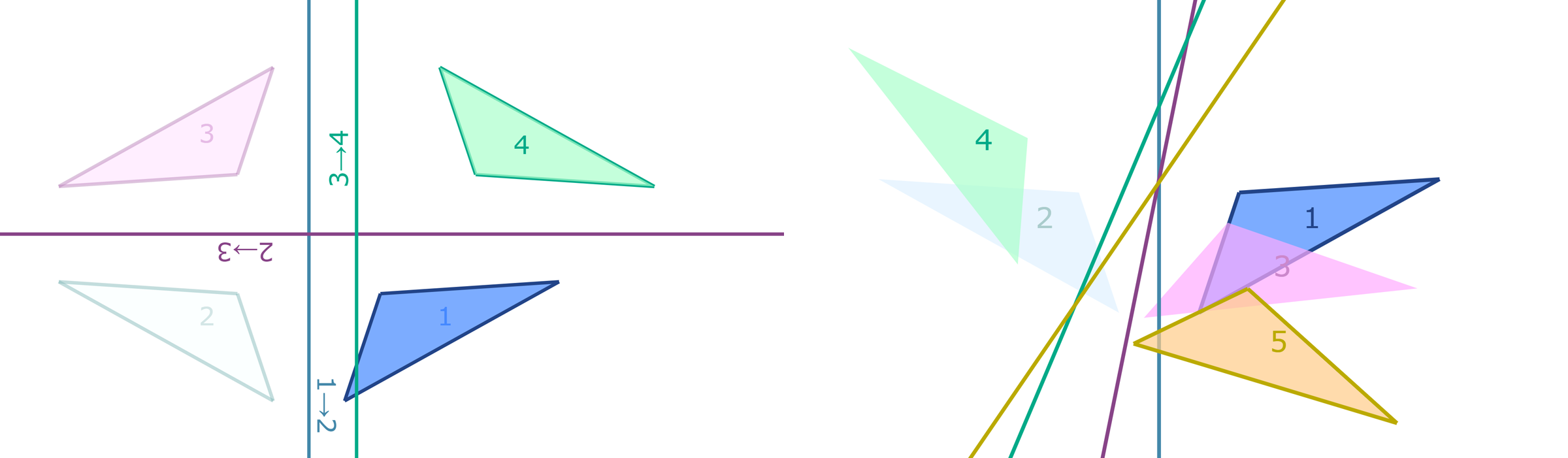}
  \caption{(left) Three line reflections produce a glide reflection. \Cref{ex:glide_reflection} shows that in \Eu{2} the seemingly distinct rotoreflections are also just glide reflections. (right) Four line reflections produce no new isometry, but instead again a rotation/translation.}
  \label{figure_isometries}
\end{figure}
But when four reflections are combined, we can factorize these  such that two adjacent reflections are identical. Two identical reflections leave the entire plane unchanged, and as a consequence every composition of four reflections in a plane can be simplified down to two reflections, as demonstrated in \cref{figure_cartan}.
\begin{figure}[tbh]
  \noindent
  \centering
  \includegraphics[width=1.0\columnwidth]{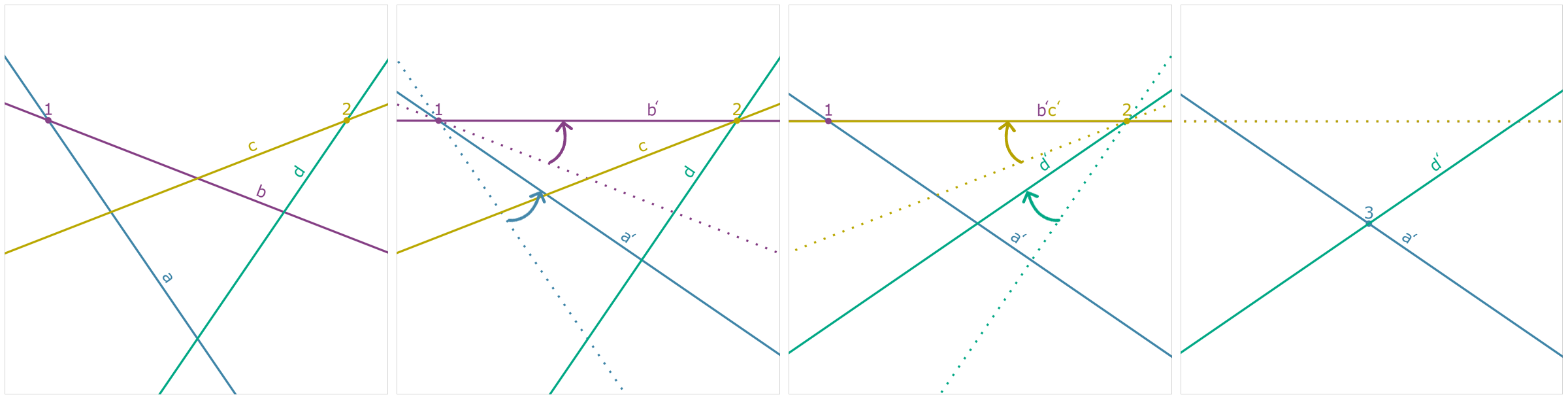}
  \caption{Given four line reflections $abcd$, the gauge symmetry allows the pairs $ab$ and $cd$ to be rotated into $a'b'$ and $c' d'$, until $b'$ and $c'$ are incident and vanish from the expression. The result is a bireflection.}
  \label{figure_cartan}
\end{figure} 

\subsection{Graded Symmetry Groups}\label{Graded Groups}
Before we continue our geometric view, it makes sense at this point to introduce the usual group theoretic notation. A group $\mathcal G$ is a set of elements and a single binary operation, written here with juxtaposition, that satisfies the following requirements:
\begin{enumerate}
\itemsep0em 
    \item For any ordered pair $a,b \in \mathcal G : ab \in \mathcal G$ (closure)
    \item The associative law holds. $\forall a,b,c \in \mathcal G : (ab)c = abc = a(bc)$
    \item There is an identity element $1 \in \mathcal G$, such that \\
    $\forall a \in \mathcal G: 1a = a = a1$ 
    \item Each non-degenerate $a \in \mathcal G$ has an inverse \\
    $a^{-1} \in \mathcal G : aa^{-1} = 1 = a^{-1}a$
\end{enumerate}
Because compositions of reflections satisfy all these conditions, 
they form a group.
In standard notation, the reflection group of oriented reflections through the origin, in a space with $p$ positive and $q$ negative dimensions, is denoted \Pin{p,q}. However, this notation can be extended to also include $r$ null dimensions, to \Pin{p,q,r}. These null dimensions can be leveraged to describe reflections not through the origin.

In an $n$ dimensional orthogonal space, any isometry can be written as a product of at most $n$ reflections. This is the famous Cartan-Dieudonné theorem:
\begin{theorem}[Cartan-Dieudonné]\label{th:cd}
Every orthogonal transformation of an $n$-dimensional space can be decomposed into at most $n$ reflections in hyperplanes.
\end{theorem}
\begin{proof}
The proof follows by induction. When $n=1$, the only isometry different from identity is clearly a point reflection through the origin. Now assume Cartan-Dieudonné holds in $n=k-1$ dimensions, and thus that any isometry in $\phi_{k-1}$ can be written as $\phi_{k-1} = \prod_{i=1}^{k-1} u_i$, where $u_i$ is a reflection. We shall now prove that it also holds for $n=k$ dimensions.

Let $\phi_k$ be an isometry different from identity in $n=k$ dimensions, fixing the origin. Then assume $v \in \mathbb{R}^k$ is a hyperplane such that $\phi_k[v] \neq v$.

Define the bisector $u_k = \overline{v - \phi_k(v)}$.
Reflection in the bisector $u_k$ maps $v$ to the same hyperplane as $\phi_k$. Hence, $\phi_{k-1} = u_k \phi_k$ acts as the identity element on~$v$: $\phi_{k-1}[v] = v$, while $\phi_{k-1}$ is an isometry on the $k-1$ dimensional subspace orthogonal to $v$. But by assumption $\phi_{k-1}$ can be written as a product of reflections: $\phi_{k-1} = \prod_{i=1}^{k-1} u_i$. 
Therefore, the isometry $\phi_k$ can be written as $\phi_k = \prod_{i=1}^k u_i$.
\end{proof}
\noindent
Hence, all isometries are graded by the number of reflections they can be factored in. This leads us to the definition of graded symmetry groups:

\begin{definition}[Graded symmetry group]
    An element $U \in \Pin{p,q,r}$ is said to be of grade $k$ if it can be factored into $k$ linearly independent reflections:
        \[ U = u_1 u_2 \cdots u_k. \]
    Because the product of a $k$-reflection with a linearly independent $l$-reflection is a $(k+l)$-reflection, Pin groups are \emph{graded symmetry groups} \cite[Chapter 3.1]{bourbaki1989algebra}.
\end{definition}
\noindent
As is customary, we write group elements using lowercase roman characters ($a,b,c$), however, we write $k$-reflections using names with $k$ characters. This aligns with our composition operator, which is written using juxtaposition. As an example, in $E(2)$ a rotation might be called $ab$, while a glide reflection might be called $abc$.

\subsection{Conjugation - Group Action}\label{Group Action}
We now turn our attention to the intuition behind the conjugation, i.e. the action of the group on itself. We know how to compose two reflections $a,b$ into a bireflection $ab$, but how do we reflect $b$ in $a$? The solution is the one we intuitively use when being asked to write upside down. We turn the page, write, and turn the page back. Indeed, we can simply use composition and this `sandwich' construction to apply one transformation to another. The expression 
$aba^{-1}$ will reflect $b$ w.r.t $a$, while $(ab)c(b^{-1}a^{-1})$ will rotate reflection $c$ with rotation $ab$, which is equivalent to $a(bcb^{-1})a^{-1}$, applying reflection $a$ after reflection $b$ to $c$.
Because reflections do not commute, the order of application is important. Additionally, because applying the same reflection twice is the same as doing nothing, reflections satisfy $a^2 = \pm 1$, and hence $a^{-1} = a / a^2$.
(In Euclidean spaces $a^2 = 1$, whereas in general pseudo-Euclidean spaces we can also encounter reflections satisfying $a^2 = -1$.)
In \cref{figure_action} we use gauge symmetry to reflect $b$ w.r.t~$a$, similar to the process of  \cref{figure_freedom,figure_cartan}.
(Note that if orientation is relevant, we have to be a bit more careful, see \cref{sec:double_cover,sec:conjugation}.)
\begin{figure}[h]
  \noindent
  \centering
  \includegraphics[width=1.0\columnwidth]{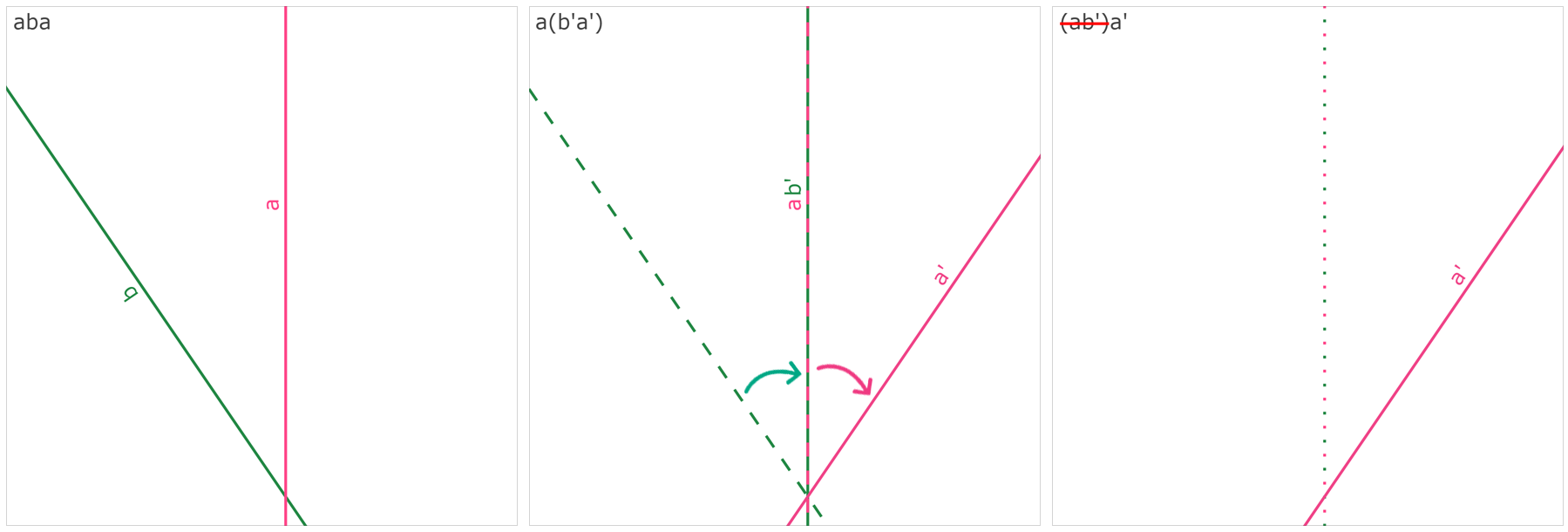}
  \caption{Reflecting the reflection $b$ in the reflection $a$. In the expression $aba$, we rotate $ba$ into $b'a'$ such that $b'=a$ and $ab'=a^2=1$. The remaining $a'$ is the reflection of $b$ w.r.t. $a$.}
  \label{figure_action}
\end{figure}

\subsection{Beyond \Eu{2}}
The Euclidean group of the plane is a great place to build intuition, but we would like to conclude the geometric introduction by showing how this mindset extends to other dimensionalities and groups.

\subsubsection{Pseudo-Euclidean \Eu{p, q}}
The ideas introduced above in \Eu{2} apply identically to higher dimensions and signatures,
where reflections are performed in hyperplanes in general pseudo-Euclidean groups $\Eu{p, q} \cong \Pin{p,q,1} / \{ \pm 1\}$.
The orthogonal basis hyperplanes of a pseudo-Euclidean group are $p$ positive hyperplanes and $q$ negative hyperplanes through the origin, and a single null hyperplane to represent the hyperplane at infinity.
In three dimensional space this is of course the important Euclidean group of rigid body motions. 
For general hyperplanes in $n = p+q$ pseudo-Euclidean dimensions, $n+1$ reflections create closure, providing an increasingly rich set of transformations as the number of dimensions increases.
However, as we shall prove in \cref{sec:invariant_decomposition}, these transformations can always be understood as a product of commuting bireflections acting in a plane and at most a single reflection. Such transformations are called \emph{simple}. 
Therefore, our intuitions from smaller dimensions remain valid in higher dimensions, \emph{plane and simple}.
\begin{figure}[h]
  \noindent
  \centering
  \includegraphics[width=1.0\columnwidth]{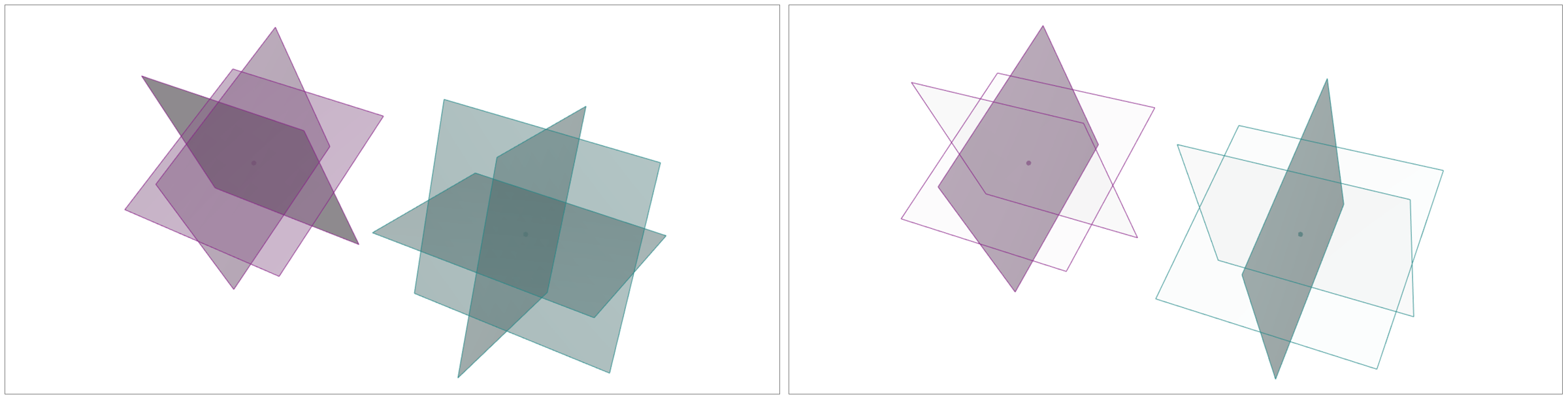}
  \caption{The composition of two point reflections in $E(3)$ produces a translation - as before we can use the gauge symmetries to rotate each set of three reflections around their intersection until identical reflections annihilate, ultimately leaving only two parallel reflections.}
  \label{figure_e3}
\end{figure}
\subsubsection{Orthogonal $\Or{p, q}$}

The Orthogonal group $O(p, q)$, which contains only reflections and rotations that keep the origin fixed, is a \emph{subgroup} of the pseudo-Euclidean group $E(p, q)$. As it is embedded in the Euclidean group, we can construct it by selecting only the reflections in (hyper)planes through the origin. In the orthogonal groups, closure is reached after exactly $n$ reflections. (This is easily verified in the orthogonal group in 2 dimensions $O(2)$, where any 3 reflections through the origin can be simplified to just one, using the gauge symmetry just like in \cref{figure_cartan}.)

\subsubsection{Conformal $\Or{p+1,q+1}$}

The Conformal group \Conformal{p,q}, whose transformations preserve local angles, is also generated from successive reflections. In this case they are inversions: `reflections' in hyperspheres (circles in 2D, spheres in 3D, ...). Realizing that a line can be seen as a circle with an infinite radius, it is clear that the pseudo-Euclidean group \Eu{p,q} is a subgroup of the conformal group \Conformal{p,q}.

\begin{figure}[h]
  \noindent
  \centering
  \includegraphics[width=1.0\columnwidth]{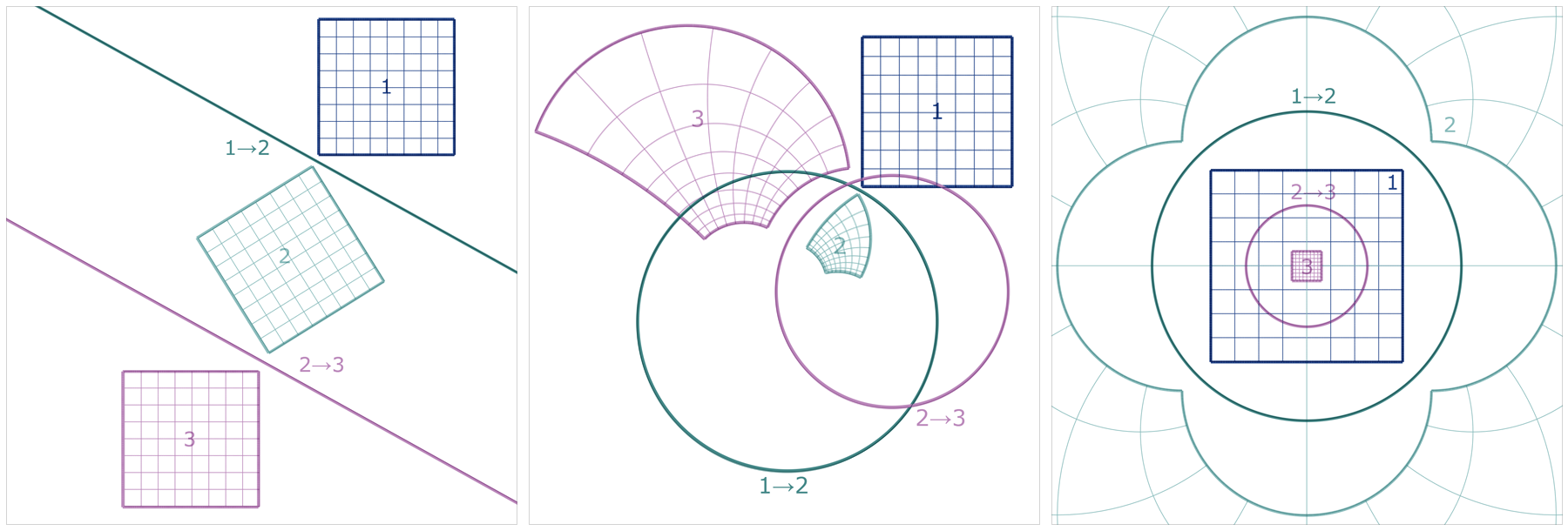}
  \caption{Reflecting grid $1$ in two circles. (left) Reflections in two parallel lines (circles with infinite radius) produce translations. (middle) Reflections in two general circles produce conformal transformations. (right) Reflections in two circles with a shared center produce dilations. }
  \label{figure_conform}
\end{figure}

\subsection{Observing orientation leads to double covers}\label{sec:double_cover}

The Pin and Spin groups, which will turn out to play an important role in the remainder of this article, are the double covers of the orthogonal and special orthogonal groups respectively. They too admit a simple geometric interpretation that fits in the same intuitive framework.

The hyperplanes (or hyperspheres) in which we reflect, are assigned an orientation, i.e. a front and back. We simply discriminate both half spaces of our fundamental reflection, a property that naturally extends to the higher grade transformations. 
As an example, \cref{fig:figure_pin} shows the four different ways in which the same rotation can be performed. We arrive at four because a rotation is a bireflection, and each of the constituent reflections can have two different orientations.
Combinations of oriented reflections produce oriented rotations, which arrive at the same final position via either the short or long path. 

The reflection group which preserves orientation is called the Pin group, denoted \Pin{p,q,r}, where $p$, $q$, and $r$ are the number of positive, negative and null dimensions respectively.
By distinguishing the orientation, the Pin group always has two distinct elements representing the same transformation, while the orthogonal group has only one. This is why \Pin{p,q,r} is a double cover of \Or{p,q,r} \cite{Poteous1969,Vaz:2016qyw}:
\begin{equation*}
    \Or{p,q,r} \cong \Pin{p,q,r} / \{\pm 1 \}.
\end{equation*}
Similarly, the group of oriented $2k$-reflections, \Spin{p,q,r}, is a double cover of the group of non-oriented $2k$-reflections, \SO{p,q,r}:
\begin{equation*}
    \SO{p,q,r} \cong \Spin{p,q,r} / \{\pm 1 \}.
\end{equation*}
Various other noteworthy double covers are listed in \cref{tab:double_cover}. 
The graded symmetry group approach makes it straightforward to preserve orientation, and so the rest of this paper will concern \Pin{p,q,r} unless otherwise specified.

\begin{figure*}[ht]
  \noindent
  \centering
  \includegraphics[width=1.0\textwidth]{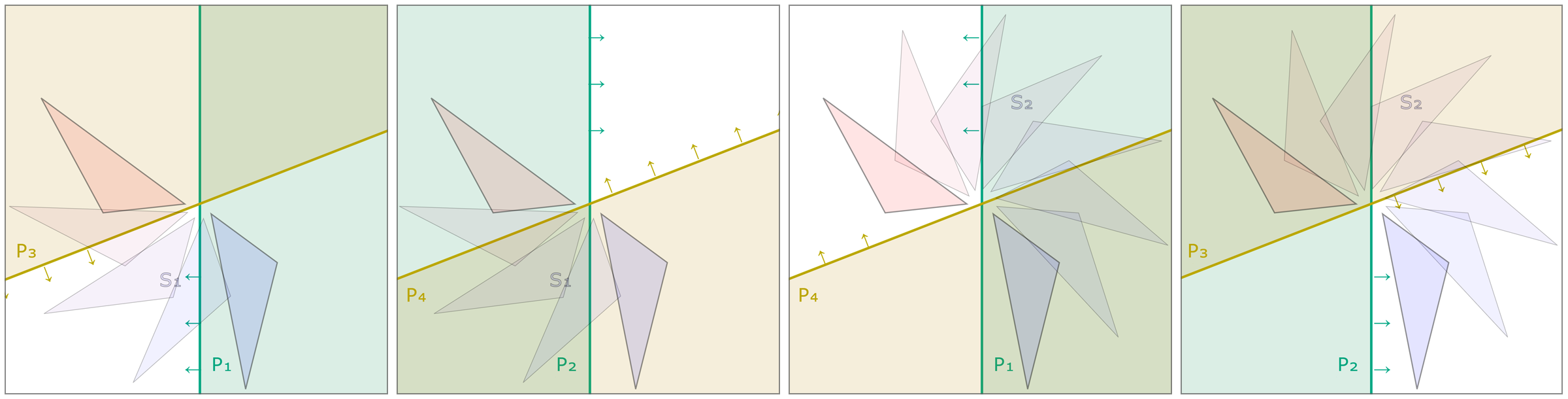}
  \caption{One rotation from $\text{SO}(2)$, that takes the blue triangle to the red one, two distinguishable elements $S_i$ from $\text{Spin}(2)$ that result in the same rotation (the two paths), and 4 oriented reflections $P_j$ from $\text{Pin}(2)$ that combine into $S_i$. Note that as before, any two reflections with the same intersection point and relative angle and orientation will produce the same elements in $\text{Spin}(2)$ and $SO(2).$}
  \label{fig:figure_pin}
\end{figure*}
\begin{table}[h]
    \centering
    \begin{tabular}{c}
    $k$-reflections\\
    \hline
        $
            \begin{array}{r @{{}\cong{}} r @{{}/{}} l}
                \Or{p,q,r} & \Pin{p,q,r} & \{\pm 1 \} \\
                \Eu{p,q} & \Pin{p,q,1} & \{\pm 1 \} \\
                \Conformal{p,q} & \Or{p\text{+}1,q\text{+}1} \, & \{\pm 1 \} \\
                & \Pin{p\text{+}1,q\text{+}1}  & \{\pm 1, \pm i\}
            \end{array}
        $
    \end{tabular}
    \begin{tabular}{c}
        $2k$-reflections\\
        \hline
        $
            \begin{array}{r @{{}\cong{}} r @{{}/{}} l}
                \SO{p,q,r} & \Spin{p,q,r} & \{\pm 1 \} \\
                \SE{p,q} & \Spin{p,q,1} & \{\pm 1 \} \\
                \SConformal{p,q} & \SO{p\text{+}1,q\text{+}1} \, & \{\pm 1 \} \\
                & \Spin{p\text{+}1,q\text{+}1}  & \{\pm 1, \pm i\}
            \end{array}
        $
    \end{tabular}
    \vspace{0.2cm}
    \caption{The (special) orthogonal, pseudo-Euclidean and conformal groups are unoriented versions of Pin groups \cite{Poteous1969,Vaz:2016qyw}.}
    \label{tab:double_cover}
\end{table}

\section{Geometric Algebra}\label{sec:geometric_algebra}

For practical calculations involving symmetry groups, the typical approach is to use matrix representations.
We would however like to consider an alternative approach, which models symmetry groups using geometric algebras (Clifford algebras) over the real numbers. 
This choice is based on our geometric intuition: combining reflections imposes a graded structure, a property shared with geometric algebras, yet hidden in the matrix approach.

\subsection{Introduction}

We start by giving a brief introduction to the geometric algebra concepts used in this paper.
For a full pedagogical introduction to geometric algebra we refer the reader to \cite{GA4Ph,GA4CS}, although the approach taken in this paper extends that taken in traditional resources.
A geometric algebra $\GR{pqr}$ of $n$ = $p+q+r$ dimensions has $p$~positive, $q$~negative, and $r$~null vectors $\e{i}$, with
    \[\e{i} \e{i} \in \{+1,-1,0\}, \quad \e{i}\e{j}=-\e{j}\e{i}. \]
A product of $k$ basis vectors is a basis $k$-blade, denoted e.g.
    \[\e{ij} = \e{i} \e{j} ,\quad \e{ijk} = \e{i}\e{j}\e{k}, \quad \ldots, \quad I := \e{1 2 \ldots n} = \e{1} \e{2} \cdots \e{n}. \]
The highest grade basis blade $I := \e{1 2 \ldots n}$ is the \emph{pseudoscalar} of $\GR{pqr}$.
These $k$-blades combine into $k$-vectors (of \emph{grade} $k$), e.g. a vector
$a = \sum_i a_i\e{i}$, a bivector $B = \sum_{i > j} B_{ij}\e{ij}$, etc.
A general \emph{multivector} $x$ is a sum of $k$-vectors:
    \[x = \expval{x} + \expval{x}_1 + \expval{x}_2 + \ldots + \expval{x}_n, \]
where $\expval{x}_k$ denotes the grade $k$ part of $x$ and $\expval{x} := \expval{x}_0$ is the scalar part.
The product of a $k$-vector $x$ and $l$-vector $y$
    \[xy = \expval{xy}_{\lvert k-l \rvert} + \expval{xy}_{\lvert k-l+2 \rvert} + \ldots + \expval{xy}_{k+l-2} + \expval{xy}_{k+l},\]
has lowest grade $\lvert k-l\rvert$ and highest grade $k+l$, that define the \emph{inner} (dot) and \emph{outer} (wedge) products
    \[x \cdot y := \expval{xy}_{\lvert k-l \rvert},\quad x \wedge y := \expval{xy}_{k+l}. \]
A general $k$-blade is the wedge product of $k$ vectors, or equivalently the product of $k$ anti-commuting/orthogonal vectors.
Because of this property, blades always square to scalars and are called \emph{simple}. E.g. $(x \wedge y)^2 = (x' y')^2 = - x'^2 y'^2 \in \mathbb{C}$, where $x' \cdot y' = 0$.

Another useful product is the \emph{commutator product}, defined by 
    \[x \times y := \tfrac{1}{2} \pqty{xy-yx}. \]
The main involution $\hat x$ of a multivector $x$ is defined by the map $\e{i} \leftrightarrow -\e{i}$. 
The reverse $\tilde x$ of a multivector $x$ is defined by reversing the order of the basis factors $\e{i}\e{j}...\e{n} \leftrightarrow \e{n}..\e{j}\e{i}$. Both of these maps reduce to grade-dependent sign changes on the original multivector.

The squared norm of any element is defined as
    \begin{equation}
        \norm{x} := \begin{cases}
            \sqrt{x\tilde{x}} & \Re\pqty{ x\tilde{x}} \geq 0 \\
            \sqrt{- x\tilde{x}} & \Re\pqty{ x\tilde{x}} < 0
        \end{cases}
    \end{equation}
Using the norm, any element $x$ for which $\norm{x} \neq 0$ can be normalized to
    \begin{equation}\label{eq:norm}
        \overline{x} := x / \norm{x}.
    \end{equation}

\subsection {Embedding reflections}
To obtain an algebraic representation of graded symmetry groups, we need to establish an embedding which maps grade 1 reflections onto grade 1 vectors, while mapping the product to the composition operator.
Because reflections square to $\pm 1$, this is an ideal task for Clifford algebras.
For the Pin groups \Pin{p,q} finding this algebraic representation is a straightforward task, as the invertible vectors $v \in \GR{pq}$, 
normalized such that $v^2 = \pm 1$, 
can directly be identified with reflections in hyperplanes through the origin. 
This is because a hyperplane through the origin, defined via a linear equation $ax + by + cz + ... = 0$, can be mapped onto a vector representing that hyperplane using
\begin{equation*}
    ax + by + cz + \ldots = 0 \rightarrow v = a\e1 + b\e2 + c\e3 + \ldots .
\end{equation*}
However, for the pseudo-Euclidean and conformal groups the process is a bit more involved. To realize a GA representation of a pseudo-Euclidean group, we need to represent general hyperplanes, not just those through the origin. This is accomplished by the embedding
\[
ax + by + cz + ... + \delta = 0 \rightarrow v = a\e1 + b\e2 + c\e3 + ... + \delta \e0,
\]
where $\delta$, the oriented offset from the origin, is associated with the null vector $\e0$ of $\GR{p,q,1}$.

For the Conformal group, it is customary to first define a Witt basis, where two null-vectors are defined to represent spheres at the origin with zero and infinite radius.
\[
n_o := \frac 1 2 (\e{-} - \e{+}), \quad
n_\infty := \e{-} + \e{+}.
\]
Here $\e{+}, \e{-}$ are the extra positive and negative basis vectors of $\GR{p+1,q+1}$ respectively. A general hypersphere with radius $\rho$ at position 
\[ x = x^1 \e1 + x^2 \e2 + x^3 \e3 + ...\]
is now mapped to a vector using\footnote{Note that this embedding is dual to the customary one used in CGA, where hyperspheres are instead represented by $(n-1)$-vectors \cite{GA4CS}}
\[
 v = n_o + x + \tfrac 1 2 (x^2 - \rho^2) n_\infty.
\]
For the remainder of this paper we will only refer to hyperplanes, in the understanding that the word hyperplane can always be replaced by hypersphere.
In all cases these embeddings produce vectors that represent the desired reflections or inversions respectively, and compose as required using the geometric product. 
With this minimal amount of setup we obtain a unified algebraic method to perform practical calculations with graded symmetry groups. 

\subsection{Conjugation - Applying reflections}\label{sec:conjugation}

With the identification of oriented hyperplanes with signed vectors in mind, a hyperplane $v$ reflects in a hyperplane $u$ as
    \begin{equation}
        v \mapsto u[v] = - u v u^{-1}.
    \end{equation}
This ensures that $u[u] = -u$, and hence a hyperplane reflected in itself changes orientation. 
For an intuitive understanding of the sandwich structure we think back to \cref{Group Action}: when asked to write upside down, one simply rotates the paper, writes, then rotates the paper back.
Because any $l$-reflection $V = v_1 v_2 \cdots v_l$ should transform covariantly under a $k$-reflection $U = u_1 u_2 \cdots u_k$, i.e.
    \[ U[V] = U[v_1] U[v_2] \cdots U[v_l], \]
the transformation of $V$ under $U$ is
    \begin{equation}\label{eq:conjugation}
        V \mapsto U[V] = (-1)^{kl} U V U^{-1}.
    \end{equation}
The term $(-1)^{kl}$ ensures the correct orientation. 
We recognize this transformation law as that of the twisted Clifford–Lipschitz group \cite[Chapter 5.2]{Vaz:2016qyw}.
It is important to note that from an algebraic point of view, the term $(-1)^{kl}$ is not required; the twisted and non-twisted Clifford–Lipschitz groups are isomorphic \cite[Chapter 5.2]{Vaz:2016qyw}. However, the geometric demand for the correct orientation \emph{does} force the inclusion of the minus signs.

By solving $U[I] = \det(U) I$, the determinant of the transformation $U$ is found to be
    \[ \det(U) = (-1)^k. \]
Hence, if $U$ is odd it inverts handedness, whereas an even $U$ preserves handedness.

\subsection{Simple Rotors}\label{sec:continuous_transformations}
The composition of two reflections $u$ and $v$ produces the \emph {simple} bireflection $uv$, representing a continuous transformation with twice the separation $\sigma$ from $v$ to $u$, as shown in \cref{fig:figure_pin}.
When applying the bireflection $uv$ to $v$, the result is $(uv)v(uv)^{-1} = u v u^{-1}$, and hence $v$ is effectively reflected in $u$ such as displayed in \cref{figure_action}.
If $uv$ is applied $x$ times, this rotates $v$ by $2 x \sigma$, and is the same as applying $(uv)^x$ once. 
How is this to be extended to $x \in \mathbb{R}$, such that rotations over an arbitrary separation can be performed?
The first step might be to construct $\sqrt{uv} = (uv)^{1/2}$, which will rotate $v$ directly onto $u$. To do so, observe that reflection in a bisector
    \begin{equation*}
        w_\pm = \overline{u \pm v},
    \end{equation*}
where the bar denotes normalization as defined in \cref{eq:norm},
already has the effect of mapping $v$ onto $\pm u$. Specifically,
    \begin{align*}
        - w_{\pm} v w_{\pm}^{-1} &= - \frac{(u \pm v) v (u \pm v)}{(u \pm v)^2} \\
        &= - \frac{(u \pm v) (\pm v + v^2 u^{-1}) (\pm u)}{(u \pm v)^2} = \mp u,
    \end{align*}
where we have used that $u^2 = v^2$ to find $v^2 u^{-1} = u$.
When this is followed by a reflection in $u$, or preceded by a reflection in $v$, the results are the bireflections $(\pm uv)^{\frac{1}{2}} = u w_{\pm} = w_{\pm} v$, rotating $v$ into $\pm u$:
    \begin{align*}
        (\pm uv)^{\frac{1}{2}} \, v \, (\pm uv)^{-\frac{1}{2}} &= (u w_{\pm}) v (w_{\pm}^{-1} u^{-1}) \\
        &= (w_{\pm} v) v (v^{-1} w_{\pm}^{-1}) \\
        &= \pm u.
    \end{align*}
This goes back to the notion of double cover as expressed in \cref{fig:figure_pin}, since applying $(\pm uv)^{1/2}$ twice is identical to applying $\pm uv$, both of which result in the same final state, but by rotating in the opposite direction.

Now that we have found how to perform a rotation by $x=1/2$, it is clear that a rotation over any separation can be performed by forming a bireflection
    \begin{equation*}
        u (\overline{y_1 u + y_2 v}) = (\overline{y_1 u + y_2 v}) v,
    \end{equation*}
where $y_1, y_2 \in \mathbb{R}$. Because $y_1$ and $y_2$ are linked via the normalization condition, it follows that the bireflection $R = uv$ forms a one-parameter subgroup $\Refl(x) = R^x = (uv)^x$ of rotations of $2 x \sigma$ about the intersection of $u$ and $v$.
The one-parameter subgroup 
    \begin{equation*}
        \Refl(x) = R^x = e^{x \Ln R}
    \end{equation*}
is generated by the bivector $\bl = \Ln R$. To show that $\bl$ is a bivector, we consider the derivative of the normalization condition $\Refl \widetilde{\Refl} = \Refl \Refl^{-1} = 1$:
    \begin{align*}
        \dot{\Refl} \widetilde{\Refl} + \Refl \dot{\widetilde{\Refl}} = 0.
    \end{align*}
But $\widetilde{\pqty{\Refl \dot{\widetilde{\Refl}}}} = \dot{\Refl} \widetilde{\Refl}$, and thus $\widetilde{\dot{\Refl} \widetilde{\Refl}} = - \dot{\Refl} \widetilde{\Refl}$.
Only bivectors anticommute under reversal, it follows that $\dot{\Refl} \widetilde{\Refl}$ is a bivector.\footnote{For $2k$-reflections with $2k < 6$ this argument suffices; for a general proof see \cref{liealgebra}.} Explicit calculation then gives
    \[ \dot \Refl \widetilde{\Refl} = \bl \Refl \widetilde{\Refl} = \bl, \]
and thus $\bl = \dot \Refl \widetilde{\Refl} = \Ln R$ is a bivector.
The principal logarithm $\Ln R$ is given by
    \begin{align}
        \Ln R &= \begin{cases}
            \overline{\expval{R}_2} \,  \arccosh\pqty{\expval{R}} & \expval{R}_2^2 > 0 \\
            \expval{R}_2 & \expval{R}_2^2 = 0 \\
            \overline{\expval{R}_2} \,  \arccos\pqty{\expval{R}} & \expval{R}_2^2 < 0
        \end{cases}. \label{eq:simple_log}
    \end{align}
Note that the form $\Ln R = \dot R \widetilde R$ reveals the logarithm to be exactly the derivative $\dot R$ moved back with $\widetilde R$ to the origin, as expected.
If the bireflection $R$ is a spatial rotation, i.e. $\expval{R_i}_2^2 < 0$, then \cref{eq:simple_log} ensures a full $2 \pi$ range.\footnote{Typically the two parameter $\arctantwo(y,x)$ function is invoked to maintain $2\pi$ resolution, as it does all the bookkeeping needed to determine the correct quadrant. However, all such manual bookkeeping can be avoided by using \cref{eq:simple_log}.}
Additionally, the logarithm of a rotation is by no means unique. The principal logarithm $\Ln R$ is one such logarithm, but so is
\begin{equation*}
	\ln R = \Ln R + 2 \pi m \overline{\expval{R}_2},
\end{equation*}
where $m \in \mathbb{Z}$.
The principal logarithm of $2k$-reflections will be the subject of \cref{sec:logarithm}, and \cref{eq:simple_log} is at its core.

The bireflections $R = uv$ and $ \widetilde{R} = vu$ are invariants of the rotations $\Refl(x) = R^x$:
    \begin{equation*}
        R^x R R^{-x} = R, \qquad R^x \widetilde{R} R^{-x} = \widetilde{R},
    \end{equation*}
and so is any linear combination of $uv$ and $vu$.
Of particular interest are the symmetric and anti-symmetric combinations
    \begin{equation*}
        \begin{array}{r @{{}={}} c @{{}+{}} c}
            R & \tfrac{1}{2} \pqty{R + \widetilde{R}} 
            & \tfrac{1}{2} \pqty{R - \widetilde{R}} \\
            e^{\bl} & \tfrac{1}{2} \pqty{e^{\bl} + e^{-\bl}} 
            & \tfrac{1}{2} \pqty{e^{\bl} - e^{-\bl}}.
        \end{array}
    \end{equation*}
These combinations allows us to define the generalized cosine and sine functions:
    \begin{alignat}{5}
        \co\pqty{\bl} &:= u \cdot v &&= \expval{R} &&= \tfrac{1}{2} \pqty{e^{\bl} + e^{-\bl}}, \label{eq:generalized_cos} \\
        \si\pqty{\bl} &:= u \wedge v &&= \expval{R}_2 &&= \tfrac{1}{2} \pqty{e^{\bl} - e^{-\bl}}. \label{eq:generalized_sin}
    \end{alignat}
Since $e^z = \sum_{n=0}^\infty z^n / n!$ and $\bl^2 \in \mathbb{C}$, these can be simplified to
    \begin{alignat*}{3}
        \co\pqty{\bl} &= \sum_{n=0}^\infty \tfrac{1}{(2n)!} \bl^{2n} 
        &&= \cosh(\sqrt{ \bl^2 }) \\
        \si\pqty{\bl} &= \sum_{n=0}^\infty \tfrac{1}{(2n+1)!} \bl^{2n+1} 
        &&= \bl \sinhc\pqty{\sqrt{ \bl^2 }}.
    \end{alignat*}
In the last step we introduced the $\sinhc$ function over the complex numbers:
    \begin{equation*}
        \sinhc(z) := \begin{cases} \frac{\sinh(z)}{z} & z \neq 0 \\
        1 & z = 0
        \end{cases},
    \end{equation*}
where $z \in \mathbb{C}$. It follows that $\bl \propto \si(\bl) = u \wedge v$, and hence $\bl$ is a 2-blade.
It might seem unusual that $\bl^2 \in \mathbb{C}$, while we are describing Clifford algebras over the real numbers.
Indeed it is, but we will find that the invariant decomposition of a real $k$-reflection can nonetheless result in complex simple bireflections, see e.g. \cref{ex:riesz}. 
However, this is merely a manifestation of the fundamental theorem of algebra, and so no more mysterious.

The hyperplanes $u$ and $v$ intersect in a hyperline. This hyperline is shared, and left invariant, by the whole one-parameter subgroup $\Refl(x) = (uv)^x$. Therefore we wish to associate the (hyper)line with the invariant of $\Refl(x)$, the blade $\bl \propto u \wedge v$. 
This brings us to the definition of the elements of geometry as invariants of transformations.

\section{Elements of Geometry}\label{eigen}

We would like to identify which multivectors of our geometric algebra make natural representations of the elements of geometry, such as points, lines, planes, spheres, etc.
Taking \Eu{3} as an example, reflections are associated with planes by construction.
Two planes $u$ and $v$ \emph{meet} in a line, $u \wedge v$ \cite{PGA4CS}.
This association is strengthened by the observation that the bireflection $uv$ generates a rotation around a line, and is generated by the 2-blade $\bl \propto u \wedge v$, which represents this line.
Therefore, elements of geometry are blades, i.e. outer products of hyperplanes (or equivalently the product of orthogonal hyperplanes):
    \[ \underbrace{1}_{\text{space}}, \quad \underbrace{u}_{\text{hyperplane}}, \quad \underbrace{u_1 \wedge u_2}_{\text{hyperline}}, \quad \ldots, \quad \underbrace{u_1 \wedge \ldots \wedge u_{n} \propto I}_{\text{origin}} . \]
This approach to identifying elements of geometry as blades, or equivalently as invariants of transformations, is valid in all dimensions, and importantly it includes the ideal elements at infinity as valid elements of geometry.

In pseudo-Euclidean spaces \GR{pq1}, there are two different narratives. 
Taking \Eu{2} as an example, the vectors are planes through the origin $I$. These intersect the Euclidean plane to form lines. Similarly, bivectors are lines through the origin $I$, which intersect the Euclidean plane in points. As pseudo-Euclidean spaces \GR{pq1} serve as homogenous representations of the pseudo-Euclidean plane, the natural narrative is to ignore the embedding and view vectors as hyperplanes (not through the origin) of \GR{pq}, i.e. lines in the example of \Eu{2}. This is the narrative chosen in this paper, but sometimes it can be helpful to switch narrative.

Importantly, the geometric algebra approach makes no distinction between elements of geometry and transformations, both are of multivector type, and both transform identically.
By contrast,
in the standard matrix formalism, a Pin group element $R \in \Pin{p,q,r}$ is represented by a $d$-dimensional matrix representation $D(R)$, which respects the composition law of the group:
    \begin{equation}
        D(R_1 R_2) = D(R_1) D(R_2),
    \end{equation}
where $R_1, R_2 \in \Pin{p,q,r}$. These are then used to transform vectors $\vec{x}$ of dimension $d$ in the underlying vector space:
    \begin{equation}
        \vec{x} \mapsto D(R) \vec{x}.
    \end{equation}
Every element of geometry is represented by a vector of different dimension. Considering \Eu{3} as a guiding example, both points and planes are $4$ dimensional, while lines are represented by $6$ dimensional vectors satisfying the Plücker conditions. These matrix representations are given explicitly in \cref{eq:matrix_E3}.

The matrix approach therefore creates a hard distinction between elements of geometry (vectors) on the one hand, and group elements (matrices) on the other: the former transforms under the matrix-vector product, while the later transforms under conjugation via the matrix product. 
Additionally, for a given transformation $R$, different types of geometric elements transform under different matrix representations $D(R)$ of that transformation. Even for elements of geometry with the same number of dimension, such as points and planes in \Eu{3}, the transformation matrices are not identical.
In \cref{sec:matrix_E3} the matrix representations of \Eu{3} are given, to illustrate the relationship of geometric elements to their matrix representations.

The study of these matrix representations gives rise to the rich mathematical field of representation theory, for sources specific to (particle) physics see e.g. \cite{Weinberg:1995mt,bargmann}.
However, the graded symmetry group approach presents an alternative way of studying representation theory via $k$-blades.

\paragraph{Summary: Plane and simple}\label{plane_and_simple}

The association of vectors with hyperplanes, and not with points, has several advantages. 
Most importantly, it allows the graded view of reflections and their compositions.
Using again $\Eu{3}$ to illustrate, vectors are naturally associated with the plane they reflect in. 
When a vector $v$ represents a plane, the conjugation of an element $X$ with $\Refl_v = v$ represents the associated reflection of $X$. 
The composition of two orthogonal (or parallel) reflections in planes is associated with the (possibly ideal) line shared by both planes. 
As a result, when a bivector $b$ represents a line, it can be exponentiated to generate the family of bireflections $\Refl_b = e^{b}$ that leave that line invariant, i.e. around that line. 
Similarly, the composition of three orthogonal reflections produces a trivector $t$ that is naturally associated with the (possibly ideal) point where the three planes intersect.
And again this multivector represents at the same time the unique point reflection $\Refl_t = t$ that leaves that point invariant.
We call this identification of vectors with hyperplanes the \emph{plane based view}.

Identifying the geometric elements this way may seem counter intuitive at first, but the strong link between elements and their associated transformations greatly simplifies many applications. The bivectors we have considered above are all constructed as the product of two orthogonal planes, or equivalently as the outer product of two arbitrary non-identical planes. Such bivectors are called \emph{simple}. However, not all bivectors are simple. Recall that $\Eu{3}$ allows us to combine up to four reflections: the screw motions, which leave a set of two orthogonal lines invariant. (The screw axis and an orthogonal infinite line around it.) For such a transformation the associated invariant will be a non-simple bivector, namely a linear combination of the screw axis line and its orthogonal ideal line. As we will show, any non-simple bivector in an $n$-dimensional space can be decomposed into at most $\floor{n / 2}$ commuting simple bivectors. A novel procedure to do so is outlined in the following sections and is key to efficient calculation of the exponential map and various other multivector functions.

\section{Spin Groups and Algebras}\label{liealgebra}

Any composition of $m$ reflections $U = u_1 u_2 \cdots u_m$ is an element of the Lie group \Pin{p,q,r}, where $u_i^2 = \pm 1$, depending on the metric. 
We already saw in \cref{sec:continuous_transformations} that a bireflection $R_i = u_i v_i$ can be raised to the power $t \in \mathbb{R}$, and determines a one-parameter subgroup
    \[ \Refl_i(t) = R_i^t = e^{t \Ln R_i}, \]
with the geometric interpretation of a rotation, translation, or hyperbolic rotation (boost).
While the bireflection $R_i = u_i v_i$ is an element of the Lie group \Spin{p,q,r}, the bivector $\Ln R_i$ is an element of the Lie algebra \spin{p,q,r}.
Similarly, any $2k$-reflection $R = \prod_{i=1}^k R_i$ determines a one-parameter subgroup
    \begin{align*}
        \Refl(t) &= R^t =  e^{t \Ln R} \\
        &= e^{t \Ln R_1} e^{t \Ln R_2} \cdots e^{t \Ln R_k},
    \end{align*}
which is a product of $k$ rotations, translations, or boosts.
As we shall prove in \cref{th:bivector_generator}, the generator $\Ln{R}$ of the $2k$-reflection $R$ is still a bivector, and thus the Lie algebra \spin{p,q,r} is a bivector algebra \cite{LGasSG}.
In general the simple bivectors $\bl_i := \Ln R_i$ do not commute, and thus $\Ln R$ is given by the Baker–Campbell–Hausdorff formula \cite{GA4Ph,hall2003lie}:
    \begin{align*}
        \Ln{R} &=
        \sum_{i=1}^k \bl_i + \sum_{i < j} \bl_i \times \bl_j + \order{\bl_i \bl_j \bl_k},
    \end{align*}
where $\order{\bl_i \bl_j \bl_k}$ contains higher order commutators.
However, as we will prove in \cref{cor:simple_rotors}, there always exists a factorization of $R$ into $R = (u_1' v_1') (u_2' v_2') \cdots (u_k' v_k')$, such that the bireflections $u_i' v_i'$ are mutually commuting. 
In terms of these mutually commuting bireflections the principal logarithm is just
    \begin{align*}
        \Ln{R} &= \sum_{i=1}^k \Ln\pqty{u_i' v_i'}.
    \end{align*}
Any element $R \in \Spin{p,q,r}$ can therefore be understood as a product of \emph{simple} bireflections $u_i' v_i' \in \Spin{p,q,r}$, each of which follows the generalized Euler's formula
    \begin{equation*}
        u_i' v_i' = e^{\bl_i'} = \co\pqty{\bl_i'} + \si\pqty{\bl_i'}.
    \end{equation*}
In \cref{sec:factor_even} we discuss how the factorization into bireflections $u_i' v_i'$ can be performed, after which \cref{sec:logarithm} discusses how the principal logarithm $\Ln{R}$ can be found explicitly.
Since all classical Lie groups are isomorphic to Spin groups \cite{LGasSG}, this factorization is expected to be applicable to all classical Lie groups, although the scope of the current work is limited to Pin groups.

We conclude this section by proving that $\Ln{R}$ is always a bivector.

\begin{theorem}\label{th:bivector_generator}
    Any $2k$-reflection $R = (u_1 v_1) (u_2 v_2) \cdots (u_k v_k)$ is generated by a bivector.
\end{theorem}
\begin{proof}
    A $2k$-reflection $R=(u_1 v_1) (u_2 v_2) \cdots (u_k v_k)$ determines a one-parameter subgroup $\Refl(t) := R^t$, which satisfies the normalization condition $\Refl \widetilde{\Refl} = 1$, where the $t$ dependence of $\Refl$ has been suppressed to improve readability. Differentiation of this normalization condition gives
\begin{equation*}
    \dot{\Refl} \widetilde{\Refl} + \Refl \dot{\widetilde{\Refl}} = 0.
\end{equation*}
But because $\widetilde{\pqty{\Refl \dot{\widetilde{\Refl}}}} = \widetilde{\dot{\widetilde{\Refl}}} \widetilde{\Refl} = \dot{\Refl} \widetilde{\Refl}$, it follows that $\widetilde{\pqty{\Refl \dot{\widetilde{\Refl}}}} = - \pqty{\Refl \dot{\widetilde{\Refl}}}$.
Therefore $\Refl \dot{ \widetilde{\Refl}}$ swaps sign under reversion, and can only contain terms of grade $2 + 4m$ for $m \in \mathbb{N}$.
We will now prove that only $m=0$ is allowed.
Consider the action of $\Refl$ on a vector $u$:
\[ u_\Refl := \Refl u \widetilde{\Refl}. \]
Taking the derivative yields
\begin{align*}
	\dot{u}_\Refl &= \dot{\Refl} u \widetilde{\Refl} + \Refl u \dot{\widetilde{\Refl}} \\
	&= \dot{\Refl}\widetilde{\Refl} u_\Refl - u_\Refl \dot{\Refl}\widetilde{\Refl} \\
	&= 2 \pqty{\dot{\Refl}\widetilde{\Refl}} \times u_\Refl.
\end{align*}
The left hand side of this equation is a vector, and thus so is the right hand side. But the product of an $r$-vector and a $1$-vector results in an $(r-1)$-vector and an $(r+1)$-vector. So the only way for the right hand side to be a $1$-vector, is if $r=2$.
This establishes that
    \begin{equation*}
        \Ln{R} = \dot{\Refl}\widetilde{\Refl}
    \end{equation*}
is a bivector.

\end{proof}

\section{Invariant decomposition}\label{sec:invariant_decomposition}
Any bivector $\B$ in a geometric algebra 
with dimension $n$,
can be decomposed into at most $k=\floor{n/2}$ commuting orthogonal 2-blades, as was previously conjectured or partially derived by \cite{Riesz1993,GeometricCalculus}:
\begin{equation*}
	\B = \bl_1 + \bl_2 + \ldots + \bl_k,
\end{equation*}
where $\bl_i \bl_j = \bl_i \wedge \bl_j$ and thus $\bl_i \cdot \bl_j = \bl_i \times \bl_j = 0$, and each $\bl_i$ squares to a scalar $\bl_i^2 = \lambda_i$.
The $\bl_i$ are found by solving the characteristic polynomial
    \begin{align}
        0 &= \pqty{\bl_1 - \bl_i} \pqty{\bl_2 - \bl_i} \cdots \pqty{\bl_k - \bl_i}
        = \sum_{m=0}^k (-\bl_i)^{k - m} \W_m, \label{eq:blade_poly}
    \end{align}
where 
\begin{align}
	\W_{m} :=& \frac{1}{m!}\expval{\B^{m}}_{2 m} = \frac{1}{m!} \underbrace{\B \wedge \B \wedge \ldots \wedge \B}_{m} \\
	=&
	\sum_{1 \leq i_1 < i_2 < \ldots < i_m \leq k} \bl_{i_1}\bl_{i_2} \cdots \bl_{i_m}.
	\label{eq:Dm}
\end{align}
Defining $r = \floor{k/2}$, \cref{eq:blade_poly} has solutions
    \begin{align}
		\bl_i &= \begin{cases}
		    \dfrac{\lambda_i^{r} \W_0 + \lambda_i^{r-1} \W_2 + \ldots + \W_k}{\lambda_i^{r-1} \W_1 + \lambda_i^{r-2} \W_3 + \ldots + \W_{k-1}} & $k$ \text{ even} \\[1.5em]
		    \dfrac{\lambda_i^{r} \W_1 + \lambda_i^{r-1} \W_3 + \ldots + \W_k}{\lambda_i^{r}  \W_0 + \lambda_i^{r-1} \W_2 + \ldots  + \W_{k-1}} & $k$ \text{ odd}
		\end{cases},
		\label{eq:bivector_split_preview}
	\end{align}
for $\lambda_i$ with algebraic multiplicity of 1, which we shall prove in \cref{th:invariant_decomposition}.

The distinction between even and odd $k$ is only important when $\lambda_i = 0$; when $\lambda_i \neq 0$, $\bl_i^{-1} = \bl_i / \lambda_i$ can be used to show that the two forms are identical.
However, the limit of $\lambda_i \to 0$ is the same for both cases:
    \[ \lim_{\lambda_i \to 0} \bl_i = \frac{W_k}{W_{k-1}}. \]
This provides a quick method to calculate the null bivector in pseudo-Euclidean spaces.
The series in the numerator and denominator terminate after at most $k$ wedge products of $\B$.
In order to calculate the values of $\lambda_i = \bl_i^2$, \cref{eq:bivector_split_preview} is squared and rearranged, giving the polynomialmsp
    \begin{align}
		0 &= \sum_{m=0}^{k} \expval{\W_{m}^2}_0 (- \lambda_i)^{k-m} \label{eq:roots_preview} \\ 
		&= \pqty{\bl_1^2 - \lambda_i} \pqty{\bl_2^2 - \lambda_i} \cdots \pqty{\bl_k^2 - \lambda_i}. \notag
	\end{align}
Thus, the values of $\lambda_i$ are the roots of \cref{eq:roots_preview}, after which \cref{eq:bivector_split_preview} can be used to find the blades $\bl_i$. \Cref{eq:bivector_split_preview} is valid for all $\lambda_i \in \mathbb{C}$, including $\lambda_i = 0$.

We will first give some examples in small algebras to clarify the algorithm.
This section is then concluded by proving the invariant decomposition \cref{eq:bivector_split_preview} in \cref{th:invariant_decomposition}, and proving the polynomial \cref{eq:roots_preview} in \cref{th:roots}. 

\begin{example}[Invariant decomposition in STA, 3DCGA, 3DPGA.]

Consider a non-simple bivector $\B = \bl_1 + \bl_2$ in a geometric algebra with $n < 6$, such as those encountered in Spacetime Algebra (STA) \cite{GA4Ph,STA}, 3DPGA \cite{GunnThesis,PGA4CS}, or 3DCGA \cite{GA4CS}. Solving \cref{eq:blade_poly} gives
    \begin{align*}
        0 &= \pqty{\bl_1 - \bl_i} \pqty{\bl_2 - \bl_i} \\
        &= \bl_1 \bl_2 - \bl_i (\bl_1 + \bl_2) + \lambda_i \notag \\
        &= \tfrac{1}{2} \B \wedge \B - \bl_i \B + \lambda_i, \notag 
    \end{align*}
and thus
    \begin{align}
        \bl_i = \frac{\lambda_i + \tfrac{1}{2} \B \wedge \B}{\B}.
        \label{eq:split_k2}
    \end{align}
The values of $\lambda_i$ can be obtained after squaring the expression for $\bl_i$, resulting in the polynomial
    \begin{align}
        0 &= \lambda_i^2 - \lambda_i \B \cdot \B + \tfrac{1}{4}(\B \wedge \B)^2.
        \label{eq:roots_k2}
    \end{align}
The $\lambda_i$ are the roots of \cref{eq:roots_k2}:
    \begin{equation}
        \lambda_{1,2} = \tfrac{1}{2} \B \cdot \B \pm \tfrac{1}{2} \sqrt{(\B \cdot \B)^2 - (\B \wedge \B)^2}.
        \label{eq:roots_5d}
    \end{equation}
Depending on the sign of the discriminant $\Delta := (\B \cdot \B)^2 - (\B \wedge \B)^2$, this can have either real or complex solutions.

Since
    \begin{equation}
        \B^{-1} = \frac{\bl_1 - \bl_2}{\lambda_1 - \lambda_2} = \frac{\bl_1 - \bl_2}{\sqrt{\Delta}},
    \end{equation}
\cref{eq:split_k2} is valid iff $\Delta \neq 0$. 
The only potentially problematic case occurs when the discriminant $\Delta = 0$, and thus $\lambda_1 = \lambda_2$, as this implies that $\B^{-1}$ does not exist and therefore \cref{eq:split_k2} is not valid.
However, it is easy to verify from $\B = \bl_1 + \bl_2$, that in this case $\B^3 = 4 \lambda_1 \B$,
and thus the exponent of $\B$ is still well behaved. 
Additionally, although \cref{eq:split_k2} can no longer be used when $\Delta = 0$, any 2-blade $\bl_i$ which satisfies $\bl_i \B = \lambda_i + \tfrac{1}{2} \B \wedge \B$ could be used, if a split into blades is still required. 
\end{example}

\begin{example}[Mozzi–Chasles' theorem in 3DPGA]\label{ex:chasles_thm}
The famous Mozzi–Chasles' theorem is the 3DPGA (\GR{3,0,1}) case of the invariant decomposition.
The theorem states that the most general rigid body motion in 3D is a screw motion: a rotation about an axis, either followed or preceded by a translation orthogonal to that axis \cite{mozzi,chasles}. This is depicted in \cref{fig:chasles}.
\begin{figure*}
    \centering
    \includegraphics[width=\textwidth]{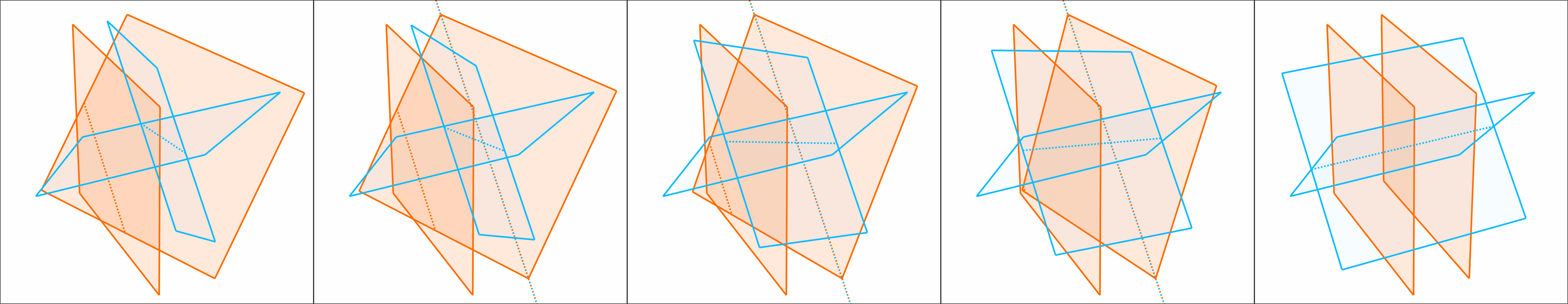}
    \caption{Mozzi-Chasles' theorem: any $4$-reflection in $\Eu{3}$ can be decomposed into two commuting bireflections: a rotation and an orthogonal translation. The $4$-reflection has $3$ gauge degrees of freedom, indicated by the dashed blue, orange, and orange-blue lines. For simplicity the panels display a rotation around the orange-blue line, while the other gauge degrees of freedom have been preset. The final result is a rotation around the blue line, and a translation along it.}
    \label{fig:chasles}
\end{figure*}

Since any handedness preserving transformation $R$ is generated by a bivector, we consider the invariant decomposition of a bivector $\B$.
Because in \GR{3,0,1} $(\B \wedge \B)^2 = 0$, 
the values of $\bl_i^2 = \lambda_{i}$ given by \cref{eq:roots_5d} are
\begin{equation*}
        \lambda_{1,2} = \tfrac{1}{2} \B \cdot \B \pm \tfrac{1}{2} \B \cdot \B.
    \end{equation*}
Consequently, $\lambda_1 = \B \cdot \B \leq 0$, and $\lambda_2 = 0$. Therefore, using \cref{eq:split_k2},
\[ \bl_2 = \frac{\B \wedge \B}{2 \B}, \quad \bl_1 = \B - \bl_2. \]
The generic motion $R = e^\B$ is therefore decomposable into a commuting orthogonal rotation $e^{\bl_1}$ and translation $e^{\bl_2}$:
    \[ R = e^\B = e^{\bl_1}e^{\bl_2} = e^{\bl_1} (1 + \bl_2).  \]
Therefore, the most general handedness preserving isometry in 3D Euclidean space is indeed a screw transformation.
\end{example}

\begin{example}[Invariant decomposition in Spacetime Algebra (STA)]

In the STA (\GR{1,3}), the pseudoscalar $I = \e{1234}$ satisfies $I^2 = -1$, and thus 
    \begin{align*}
        \lambda_{1,2} &= \tfrac{1}{2} \B \cdot \B \pm \tfrac{1}{2} \sqrt{(\B \cdot \B)^2 - (\B \wedge \B)^2} \\
        &= \tfrac{1}{2} \B \cdot \B \pm \tfrac{1}{2} \sqrt{(\B \cdot \B)^2 + \abs{\B \wedge \B^2}}.
    \end{align*}
Consequently, $\sqrt{(\B \cdot \B)^2 - (\B \wedge \B)^2} \geq \B \cdot \B$, and thus $\sign(\lambda_1) = - \sign(\lambda_2)$. Therefore any Lorentz transformation $\Lambda = e^\B$ can be decomposed into a commuting boost $e^{\bl_1}$ and rotation $e^{\bl_2}$ using \cref{eq:split_k2}, as 
    \[ \Lambda = e^\B = e^{\bl_1}e^{\bl_2}. \]
\end{example}

\begin{example}[Seeming counter example]\label{ex:riesz}

An insightful \emph{seeming} counter example to the existence of an orthogonal decomposition in all spaces, due to M. Riesz \cite[Page 170]{Riesz1993}, is the space \GR{2,2}, whose basis vectors satisfy $\e{1}^2 = \e{2}^2 = -\e{3}^2 = -\e{4}^2$. Consider e.g. the bivector
    \[ \B = \tfrac{1}{2} \pqty{\e{12} + \e{14} - \e{23} - \e{34}}, \]
which squares to $\B^2 = \B \wedge \B = - \e{1234}$. From \cref{eq:roots_5d} it follows that $\lambda_{1,2} = \pm \tfrac{i}{2}$. If we decide to exclude complex solutions, then indeed no invariant decomposition can be performed, but if we carry on regardless, we find the complex simple bivectors
    \begin{align*}
        \bl_1 &= \tfrac{1}{4}\bqty{ (1-i)\e{12} + (1+i)\e{14} + (-1-i)\e{23} + (-1+i)\e{34}}, \\
        \bl_2 &= \tfrac{1}{4}\bqty{ (1 + i)\e{12} + (1 - i)\e{14} + (-1 + i)\e{23} + (-1 - i)\e{34}}.
    \end{align*}
These satisfy $\B = \bl_1 + \bl_2$, $\bl_i^2 = \lambda_i$, and $\bl_1 \bl_2 = \bl_1 \wedge \bl_2$. Therefore, all the demands on the invariant decomposition are satisfied. Consequently, bivectors in \GR{2,2} do not pose a counter example, but rather an indication that $\lambda_i$ is allowed to be complex.

\end{example}

\begin{example}[Invariant decomposition when $5 < n \leq 7$.]
In a space of $5 < n \leq 7$, any bivector $\B$ has an invariant decomposition into
    \[ \B = \bl_1 + \bl_2 + \bl_3. \]
The $\bl_i$ are given by
    \begin{equation}
        \bl_i = \frac{\lambda_i \W_1 + \W_3}{\lambda_i \W_0 +\W_2} = \frac{\lambda_i \B + \tfrac{1}{3!} \B \wedge \B \wedge \B}{\lambda_i + \tfrac{1}{2} \B \wedge \B},
    \end{equation}
where the $\lambda_i$ are the roots of
    \begin{align}
        0 &= \lambda_i^3 - \expval{\W_1^2}_0 \lambda_i^2 + \expval{\W_2^2}_0 \lambda_i - \expval{\W_3^2}_0 \\
        &= \lambda_i^3 - \B \cdot \B \; \lambda_i^2 + \tfrac{1}{4} \pqty{\B \wedge \B}^2 \; \lambda_i - \pqty{\tfrac{1}{3!}}^2\pqty{ \B \wedge \B \wedge \B}^2. \notag
    \end{align}
\end{example}
\noindent
The matrix equivalent of this decomposition for \SU{3} was published previously \cite{roelfs2021geometric}.

\begin{theorem}[Invariant decomposition] \label{th:invariant_decomposition}
	Assuming all $\lambda_i \in \mathbb{C}$ are distinct, $\bl_i$ is given by \cref{eq:bivector_split_preview}.
\end{theorem}

\begin{proof}
First we make the ansatz that the decomposition of $\B$ into at most $k$ orthogonal 2-blades exists, to find expression \cref{eq:bivector_split_preview} for $\bl_i$, which satisfies $\B = \sum_{i=1}^k \bl_i$ by construction. Then we prove that $\bl_i \times \bl_j = 0$, thereby justifying the ansatz.

In order to prove \cref{th:invariant_decomposition}, we will prove that $\bl_i D_i = N_i$, with $N_i$ and $D_i$ the numerator and denominator of the relevant case of \cref{eq:bivector_split_preview}. Without loss of generality, let us consider $\bl_i = \bl_1$. 
The $\W_m$ satisfy the recursive relationship
    \begin{equation}
        \bl_1 \wedge \W_{m} + \tfrac{1}{\lambda_1} \bl_1 \cdot \W_{m+2} = \W_{m+1},
        \label{eq:Dm_recursion}
    \end{equation}
which is straightforwardly verified using
    \begin{align*}
        \bl_1 \W_m &= \bl_1 \cdot \W_m + \bl_1 \wedge \W_m \\
        &= \lambda_1 \sum_{1 < i_2 < \ldots < i_m} \bl_{i_2} \cdots \bl_{i_m} + \bl_1 \sum_{1 < i_1 < i_2 < \ldots < i_m} \bl_{i_1}\bl_{i_2} \cdots \bl_{i_m}.
        \label{eq:bl_Dm}
    \end{align*}
Careful evaluation of $\lambda_i \to 0$ shows that \cref{eq:Dm_recursion} also holds in this limit.
With the recursive relationship of \cref{eq:Dm_recursion}, and the realization that $\bl_i \wedge \W_{k-1} = \W_k$, the proof of \cref{th:invariant_decomposition} is immediate. 
For odd $k$,
    \begin{align*}
        N_i &= \lambda_i^{r} \W_1 + \lambda_i^{r-1} \W_3 + \lambda_i^{r-2} \W_5 + \ldots + \W_k \\
        D_i &= \lambda_i^{r} \W_0 + \lambda_i^{r-1} \W_2 + \lambda_i^{r-2} \W_4 + \ldots  + \W_{k-1},
    \end{align*}
and thus
    \begin{align*}
        \bl_1 D_1 
        &= \lambda_1^{r} \pqty{\bl_1 \W_0 + \tfrac{1}{\lambda_1} \bl_1 \cdot \W_2} + \lambda_1^{r-1}\pqty{\bl_1 \wedge \W_2 + \tfrac{1}{\lambda_1} \bl_1 \cdot \W_4} \\
        &\quad + \ldots + \bl_1 \wedge \W_{k-1} \\
        &= \lambda_1^{r} \W_1 + \lambda_1^{r-1} \W_3 + \ldots + \W_k = N_1 \notag
    \end{align*}
The proof for even $k$ follows, \emph{mutatis mutandis}, along the same lines.

This proves equation \cref{eq:bivector_split_preview}, assuming the $\bl_i$ commute. 
To prove this last statement, and thereby justify the ansatz, we need only to show that $\W_i \times \W_j = 0$. However, this follows directly if we use that $\B \times \B = 0$, and thus
    \begin{align*}
        0 &= \B^i \times \B^j \\
        &= \Big(\expval{\B^i}_0 + \expval{\B^i}_2 + \ldots + \expval{\B^i}_{2i}\Big) \times \pqty{\expval{\B^i}_0 + \expval{\B^i}_2 + \ldots  + \expval{\B^j}_{2j}}.
    \end{align*}
Since this holds grade by grade, we find for the highest grade term $\expval{\B^i}_{2i} \times \expval{\B^j}_{2j} = 0$, and thus $\W_i \times \W_j = 0$, from which it follows trivially that $\bl_i \times \bl_j = 0$.
\end{proof}

\begin{corollary} \label{cor:simple_rotors}
    Any $2k$-reflection $R = (u_1 v_2) (u_2 v_2) \cdots (u_k v_k)$ can be factored into commuting bireflections:
        \begin{equation*}
            R = (u_1' v_2') (u_2' v_2') \cdots (u_k' v_k'),
        \end{equation*}
    such that $(u_i' v_i') \times (u_j' v_j') = 0$. This follows immediately because any $2k$-reflection is generated by a bivector $\Ln R$ (\cref{th:bivector_generator}), and any bivector can be split into commuting 2-blades (\cref{th:invariant_decomposition}).
\end{corollary}

\begin{corollary} \label{cor:commuting_bivectors}
    For any value of $\lambda_i$, the bivector $\bl_i$ calculated using \cref{th:invariant_decomposition} will commute with $\B$, or with any other $\bl_j$. Only for specific values of $\lambda_i$ is $\bl_i$ also simple, which leads to \cref{th:roots}.
\end{corollary}

\begin{theorem}\label{th:roots}
	The $\lambda_i = \bl_i^2$ are the roots of \cref{eq:roots_preview} \cite[Eq. (4.14)]{GeometricCalculus}.
\end{theorem}
\begin{proof}
    Squaring \cref{eq:Dm} directly gives
    \begin{align*}
        \expval{\W_{m}^2}_0 &= \sum_{1 \leq i_1 < i_2 < \ldots < i_m \leq k} \bl_{i_1}^2 \bl_{i_2}^2 \cdots \bl_{i_m}^2 \\
        &= \sum_{1 \leq i_1 < i_2 < \ldots < i_m \leq k} \lambda_{i_1} \lambda_{i_2} \cdots \lambda_{i_m},
    \end{align*}
    and thus the equality \cref{eq:roots_preview} follows immediately.
\end{proof}

\noindent
To perform the invariant decomposition, the roots $\lambda_i$ are first determined using \cref{th:roots}, after which \cref{th:invariant_decomposition} can be used to find the corresponding  2-blades $\bl_i$.
If precisely of the $\lambda_i$ equals $0$, then $\bl_i$ can alternatively be obtained as
\begin{equation}
	\bl_i = \B - \sum_{j \neq i} \bl_j.
\end{equation}
This could simplify any implementation, as the distinction between even and odd $k$ can be dropped.
The method presented here extends previously published methods \cite{GA4Ph,GunnThesis,DorstDecomposition} to geometric algebras of  arbitrary metric and dimension, for all unique $\lambda_i \in \mathbb{C}$. 

\section{Exponential of a bivector}\label{sec:exponential}

Using the invariant decomposition of \cref{th:invariant_decomposition}, the exponential of a bivector $\B$ follows straightforwardly after performing the decomposition of $\B$ into~$\{ \bl_i \}_{i=1}^k$,
using which a group element $R = \exp\bqty{\B}$ can be written as
\begin{align}
	R &= e^{\B} = e^{\bl_1} e^{\bl_2} \cdots e^{\bl_k} \label{eq:group_element_mat} \\
	&= \prod_{i=1}^k \bqty{\co(\bl_i) + \si(\bl_i)}.
\end{align}
where $\co(\bl_i)$ and $\si(\bl_i)$ were previously defined in \cref{eq:generalized_cos,eq:generalized_sin}.
It follows that the ${\bl}_i$ span a commuting orthogonal basis for $R$: 
    \[ \{ 1, \; {\bl}_i, \; {\bl}_{ij}, \; \ldots, \; {\bl}_{12\ldots k} \}, \]
where $\bl_{ij} := \bl_i \wedge \bl_j$. These basis elements satisfy
    \[ \bl_i \bl_j = \bl_j \bl_i =  g_{ij} + \bl_i \wedge \bl_j, \]
where $g_{ij} = \bl_i \cdot \bl_j = \text{diag}(\lambda_1, \lambda_2, \dots, \lambda_k)$ is the effective metric.
Any element in this basis is invariant under each of the $R_i = e^{\bl_i}$, a special case is $\B$ itself.

With this observation in mind, we return to the gauge degrees of freedom. In \cref{E2} we have intuitively seen that a $2k$-reflection has $2k-1$ gauge degrees of freedom.
However, the commutativity imposed by the invariant decomposition restricts the number of degrees of freedom to $k$, corresponding to the number of commuting bireflections it contains. Given an even $2k$-reflection $R = R_1 R_2 \cdots R_k$, each bireflection $R_i = u_i v_i$ is specified by the reflections $u_i$ and $v_i$. However, these are not unique: $u_i$ and $v_i$ can be freely rotated around their intersection.
These rotations are determined by the one parameter subgroup $\Refl_i(\theta_i) = R_i^{\theta_i}$, since any
    \begin{align*}
        u_i' = \Refl_i[u_i], \quad v_i' = \Refl_i[v_i],
    \end{align*}
define the same bireflection $R_i$:
    \[ u_i' v_i' = R_i^{\theta_i} u_i R_i^{-\theta_i} R_i^{\theta_i} v_i R_i^{-\theta_i} = u_i v_i = R_i. \]
The other bireflections $R_j = u_j v_j$ with $j \neq i$ do not share this gauge degree of freedom because $u_j \times R_i = v_j \times R_i = 0$, and thus
    \begin{alignat*}{3}
        u_j &= \Refl_i[u_j] =  R_i^{\theta_i} u_j R_i^{-\theta_i} &&= u_j\\
        v_j &= \Refl_i[v_j] =  R_i^{\theta_i} v_j R_i^{-\theta_i} &&= v_j.
    \end{alignat*}
A $2k$-reflection $R$ therefore has only $k$ gauge degrees of freedom $\{ \theta_i \}_{i=1}^k$.
As a result, the $k$ parameter gauge group of the reflections $\{ u_i, v_i \}_{i=1}^k$ is
\begin{equation}
	R(\theta_1, \theta_2, \ldots, \theta_k) = e^{\theta_1 \bl_1}e^{\theta_2 \bl_2}\cdots e^{\theta_k \bl_k}.
\end{equation}

\section{Tangent decomposition}\label{sec:tangent}
To define the tangent function, we first define the generalized sine and cosine series as
\begin{align}
    \si(\B) &= \tfrac{1}{2} \pqty{e^\B - e^{-\B}} = \tfrac{1}{2} \pqty{R - \widetilde{R}} \\
    \co(\B) &= \tfrac{1}{2} \pqty{e^\B + e^{-\B}} = \tfrac{1}{2} \pqty{R + \widetilde{R}}.
\end{align}
The generalization of the tangent function is then defined in terms of $\si(\B)$ and $\co(\B)$ as
    \begin{equation}
        \ta(\B) := \frac{\si(\B)}{\co(\B)} 
        = \frac{R - \widetilde{R}}{R + \widetilde{R}}.
    \end{equation}
We additionally define the bivector
    \[ \Tbl := \frac{\expval{R}_2}{\expval{R}} = \sum_{j=1}^k \ta(\bl_j), \]
where the last equality follows from
    \begin{align*}
        \expval{R} &= \prod_{i=1}^k \co(\bl_i), \quad \expval{R}_2 = \sum_{i=1}^k \si(\bl_i) \prod_{j\neq i} \co(\bl_j),
    \end{align*}
where the $\bl_i$ are given by the invariant decomposition of $\B$.
We would like to find the simple bivectors $\ta(\bl_i)$, 
which are obtained by applying the invariant decomposition \cref{sec:invariant_decomposition} to the bivector $\Tbl$.
However, \cref{th:trig} allows the quantities $\W_m$ to be expressed using various grades of $R$:
    \begin{align*}
    	\W_{m} =& \frac{1}{m!}\expval{R} \expval{\Tbl^{m}}_{2 m} = \expval{R}_{2m},
    \end{align*}
where the limit $\expval{R} \to 0$ is well-behaved.
Therefore the invariant decomposition has the solutions
    \begin{align}
         \ta(\bl_i) = \begin{cases}
		    \dfrac{\lambda_i^{r} \expval{R} + \lambda_i^{r-1} \expval{R}_4 + \ldots + \expval{R}_{2k}}{\lambda_i^{r-1} \expval{R}_2 + \lambda_i^{r-2} \expval{R}_6 + \ldots + \expval{R}_{2k-2}} & $k$ \text{ even} \\[1.5em]
		    \dfrac{\lambda_i^{r} \expval{R}_2 + \lambda_i^{r-1} \expval{R}_6 + \ldots + \expval{R}_{2k}}{\lambda_i^{r} \expval{R} + \lambda_i^{r-1} \expval{R}_4 + \ldots  + \expval{R}_{2k-2}} & $k$ \text{ odd}
		\end{cases},
		\label{eq:tan_split}
    \end{align}
where $r = \floor{k/2}$. 
To find the values of $\lambda_i$, \cref{eq:roots_preview} becomes
    \begin{align}
        0 &= \sum_{m=0}^{k} \expval{\expval{R}_{2m}^2} (- \lambda_i)^{k-m}.
    \end{align}
The interesting feature of this formulation of the invariant decomposition is that it uses all the grades of the rotor, making for an exception free experience.
We will now give some examples of the tangent decomposition, after which we conclude this section with \cref{th:trig}.

\begin{example}[Tangent decomposition in STA, 3DCGA, 3DPGA, etc.]\label{ex:STA_tangent_decomposition}
A rotor $R$ in a space with $n < 6$ has a tangent decomposition given by
    \begin{equation*}
        \ta(\bl_i) = \frac{\lambda_i \expval{R} + \expval{R}_4}{\expval{R}_2}
    \end{equation*}
where the $\lambda_i$ are the roots of
\begin{equation*}
    \lambda_i^2 \expval{R}^2 - \lambda_i \expval{R}_2 \cdot \expval{R}_2 +  \expval{R}_4^2 = 0,
\end{equation*}
which are
    \begin{align*}
        \lambda_{1,2} &= \frac{\expval{R}_2 \cdot \expval{R}_2 \pm \sqrt{\pqty{\expval{R}_2 \cdot \expval{R}_2}^2 - 4 \expval{R}^2 \expval{R}_4^2}}{2 \expval{R}^2}
    \end{align*}
\end{example}

\begin{theorem}\label{th:trig}
Given a non-simple rotor $R$ defined by \cref{eq:group_element_mat}, for $0 \leq m \leq k$, the grade $2m$ part $\expval{R}_{2m}$, is given by
\begin{align}
    \expval{R}_{2m} &= \frac{\expval{\expval{R}_2^m}_{2m}}{m! \expval{R}_0^{m-1}} \label{eq:trig_iden} \\
    &= \frac{1}{m! \expval{R}_0^{m-1}} \underbrace{\expval{R}_2 \wedge \ldots \wedge \expval{R}_2}_{m} \notag
\end{align}
\end{theorem}
\begin{proof}
    We will need to use
    \begin{alignat*}{3}
        \expval{R}_0 &= \prod_{i=1}^k \co(\bl_i), \quad & \expval{R}_{2k} &= \prod_{i=1}^k \si(\bl_i), \\
        \expval{R}_2 &= \sum_{i=1}^k \si(\bl_i) \prod_{j\neq i} \co(\bl_j), \quad  
        & \expval{R}_{2(k-1)} &= \sum_{i=1}^k \co(\bl_i) \prod_{j\neq i} \si(\bl_j).
    \end{alignat*}
    By direct computation we find the recursive relationship
    \begin{align*}
        \expval{R}_2 \wedge \expval{R}_{2(k-1)} &= k \prod_{i=1}^k \si(\bl_i)\co(\bl_i) \\
        &= k \expval{R}_0 \expval{R}_{2k},
    \end{align*}
    and therefore by induction we find
    \begin{equation*}
        \expval{R}_{2k} = \frac{\expval{R}_2 \wedge \expval{R}_{2(k-1)}}{k \expval{R}_0} = \frac{\expval{\expval{R}_2^k}_{2k}}{k! \expval{R}_0^{k-1}}.
    \end{equation*}
\end{proof}

\section{Factorization  \& Logarithm of rotors}\label{sec:factorization}

We will demonstrate that any $r$-reflection $R$ can be factored into $\ceil{\frac r 2}$ mutually commuting factors.
For a $2k$-reflection these are $k$ bireflections, for a $(2k+1)$-reflection these are $k$ mutually commuting bireflections and a reflection. 
Once the factorization of a $2k$-reflection into $k$ mutually commuting bireflections has been performed, the logarithm follows immediately.

\subsection{Factorization of a \texorpdfstring{$2k$}{2k}-reflection into bireflections}\label{sec:factor_even}
Given a $2k$-reflection $R$, we want to find the factorization 
    \[ R = R_1 R_2 \cdots R_k \]
into $k$ Euler's formulas $R_i = e^{\bl_i}$.
The tangent decomposition of $R$ yields the simple bivectors $\ta(\bl_i)$, which can be used to find the $R_i$ up to sign, since
    \begin{align}
        R_i &= 
        \overline{\bqty{1 + \ta(\bl_i)}} \label{eq:bireflection_factor} \\
        &= \abs{\co(\bl_i)} \, \bqty{1 + \ta(\bl_i)} \notag \\
        &= \abs{\co(\bl_i)} + \sign\bqty{\co(\bl_i)} \si(\bl_i) \notag 
    \end{align}
To preserve the distinction between $\pm R$, the first $k-1$ bireflections are calculated using \cref{eq:bireflection_factor}, after which the final bireflection follows from
    \begin{equation}
        R_k~=~\widetilde{R}_1 \cdots \widetilde{R}_{k-1} R.
    \end{equation}

\begin{example}[3DPGA]
    In the particular case of 3DPGA the factorization of $R$ into a rotation $R_1$ and a translation $R_2$ can be greatly simplified, because the tangent decomposition of \cref{ex:STA_tangent_decomposition} always yields $\lambda_2 = 0$ and thus
    \begin{equation*}
        \ta(\bl_2) = \frac{\expval{R}_4}{ \expval{R}_2}, \quad \ta(\bl_1) = \expval{R}_2 - \ta(\bl_2).
    \end{equation*}
    Therefore, the factors $R_1$ and $R_2$ are
    \begin{equation*}
        R_2 = 1 + \ta(\bl_2), \quad R_1 = R / R_2.
    \end{equation*}
    If $\expval{R}_4 = 0$, the rotor $R$ is simple and the factorisation trivial.
\end{example}

\begin{example}[Wigner rotation]
    The multiplication of two non-collinear boosts produces a Wigner rotation \cite{10.2307/1968551}, which can be decomposed into a mutually commuting orthogonal boost and rotation using \cref{eq:bireflection_factor}. Let $R_1 = e^{\kbl_1}$, and $R_2 = e^{\kbl_2}$ be the two original boosts. These boosts do not commute, and so $R = e^{\kbl_1} e^{\kbl_2} \neq e^{\kbl_1 + \kbl_2}$ but typically requires the Baker–Campbell–Hausdorff formula. However, the tangent decomposition allows the product $R$ to be decomposed in a mutually commuting orthogonal boost $e^{\bl_1}$ and rotation $e^{\bl_2}$. Grade by grade we find
        \begin{align*}
            \expval{R} &= \co\pqty{\kbl_1} \co\pqty{\kbl_2} + \si\pqty{\kbl_1} \cdot \si\pqty{\kbl_2}\\
            \expval{R}_2 &= \co\pqty{\kbl_1} \si\pqty{\kbl_2} + \si\pqty{\kbl_1} \co\pqty{\kbl_2} \\
            \expval{R}_4 &= \si\pqty{\kbl_1} \wedge \si\pqty{\kbl_2}
        \end{align*}
    The techniques of \cref{ex:STA_tangent_decomposition} will yield $\ta(\bl_1)$ and $\ta(\bl_2)$, after which 
    \[ R_1 = \overline{1 + \ta(\bl_1)}, \qquad R_2 = \widetilde{R}_1 R. \]
\end{example}

\subsection{Factorization of a \texorpdfstring{$(2k+1)$}{(2k+1)}-reflection}\label{sec:factor_odd}
Consider a $(2k+1)$-reflection $P$ in a geometric algebra \GR{pqr} of dimension $n=p+q+r$. 
Such an element can be factored into a mutually commuting reflection $r$ and $2k$-reflection $R$, such that 
    \begin{equation}
        P = r R = R r, \label{eq:factor_odd_reflection}
    \end{equation}
where $r = \overline{\expval{P}_1}$. This follows from the Cartan-Dieudonné \cref{th:cd}, by using that $r$ is a valid factor of $P$, and therefore $R := P / r$ is a $2k$-reflection. Because $r$ maps back to the same subspace, since
    \[ r[r] = - r, \]
$R$ is an isometry orthogonal to $r$ by Cartan-Dieudonné.
It follows that the sought-after commuting reflection $r$ and $2k$-reflection $R$ are simply
    \begin{equation}
        r = \overline{\expval{P}_1}, \qquad R = \frac{P}{r}.
    \end{equation}
This completes the factorization of a $(2k+1)$-reflection into a mutually commuting reflection and $2k$-reflection, the latter of which can be decomposed into $k$ bireflections using \cref{sec:factor_even}.

\begin{example}\label{ex:glide_reflection}
\begin{figure*}[ht]
	\noindent
	\centering
	\includegraphics[width=1.0\textwidth]{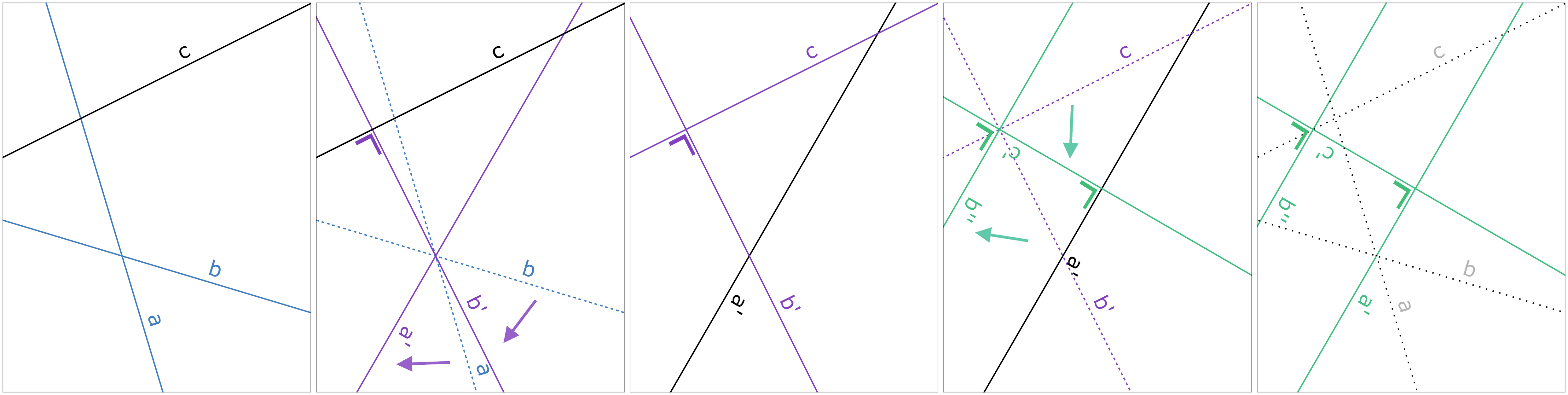}
	\caption{From left to right. (1) An arbitrary trireflection $abc$. (2) Rotation of $ab$ into $a'b'$ so that $b'\perp c$. (3) The composition is now a point reflection $b'c$ followed by a reflection $a'$. (4) Rotate $b'c$ into $b''c'$ so that $c' \perp a'$. (5) The final decomposition of $abc = a'b''c'$ reveals a commuting translation $a'b''$ and orthogonal reflection $c'$. (This example is in fact the 2D equivalent of Mozzi-Chasles' theorem - and the smallest example of the invariant decomposition.)}
	\label{figure_glide}
\end{figure*}
Any trireflection $P = a b c$ can be decomposed into a commuting reflection $r_1$ and bireflection $R = r_2 r_3$. This is depicted in \cref{figure_glide}. The reflection $r_1$ is given by
    \[ r_1 = \overline{\expval{P}_1},\]
while the bireflection is given by
    \[ R = r_2r_3 = \cfrac{P}{ r_1},\]
where the division notation is unambiguous, due to the commutation of $P$ and $r_1$.
In 2DPGA, this has as a consequence that any rotoreflection is fundamentally a commuting reflection and translation, since $\expval{P}_3$ is ideal, and thus $\expval{r_2 r_3}_2$ is ideal, making $r_2 r_3 $ a translation. 

\end{example}

\subsection{Logarithm of a rotor}\label{sec:logarithm}

With the factorization of a $2k$-reflection $R$ into mutually commuting orthogonal bireflections $\{R_1, R_2, \ldots, R_k \}$ as described in \cref{sec:factor_even} in hand, the principal logarithm is simply 
\begin{equation}
	\Ln R = \Ln R_1 + \Ln R_2 + \ldots + \Ln R_k.
\end{equation}
where $\Ln R_i$ is given by \cref{eq:simple_log}:
\begin{align*}
	\Ln{R_i} &= \begin{cases}
		\overline{\expval{R_i}_2} \arccosh\pqty{\expval{R_i}} & \expval{R_i}_2^2 > 0 \\
		\expval{R_i}_2 & \expval{R_i}_2^2 = 0
		\\
        \overline{\expval{R_i}_2} \,  \arccos\pqty{\expval{R_i}} & \expval{R_i}_2^2 < 0
	\end{cases}.
\end{align*}

When $R \to U R \widetilde{U}$ for any $U \in \Pin{p,q,r}$, $\Ln(R) \to U \Ln(R) \widetilde{U}$, as can be computed directly from the series expansion.
Thus, if the invariant decomposition of $\Ln(R)$ is known, the invariant decomposition of $U \Ln(R) \widetilde{U}$ is simply $\{ U \Ln(R_i) \widetilde{U} \}_{i=1}^k$.

\section{Clifford Representation}\label{matrix}

In this section we introduce a novel algorithm for the construction of a $2^n \times 2^n$ real matrix representation for geometric algebras \GR{pqr} of dimension $n = p+q+r$.
Although much has been published on the matrix representations of Clifford algebras, see e.g. \cite{HILE199051,Polchinski:1998rr}, this algorithm has some distinguishing features.
Firstly, it allows fine grained control over basis order, permutation and metric signature. Secondly, this representation models Clifford's geometric product not just as matrix-matrix product, but also as matrix-vector product, 
enabling efficient implementation using linear algebra software. 
Lastly, we demonstrate that this particular representation leads directly to the familiar block-diagonal form of the group action, subsuming the traditional covariant, contravariant, and adjoint representations. along the diagonal. We therefore understand the block diagonal form to be a manifestation of the grade-preserving nature of conjugation.
Owing to these particularly appealing properties, we refer to this representation as the \emph{Clifford representation}.

\subsection{A matrix representation for \texorpdfstring{\GR{pqr}}{Rpqr}}
\label{clifford_rep}
For an $n$ dimensional real Clifford algebra, a $2^n \times 2^n$ real matrix realization can be constructed starting from the following basis matrices:
\begin{equation}\label{basis}
    I = \smqty[1 & 0 \\ 0 & 1]
    ,\, I' = \smqty[1&0\\0&-1]
    ,\, P = \smqty[0 & 1 \\ 1 & 0]
    ,\, Q = \smqty[0 & 1\\-1&0]
    ,\, R = \smqty[0 & 0\\1&0]
\end{equation}
\noindent
Using the $P,Q,R$ matrices, we now assign signature matrices $S_i$ corresponding to the Clifford signature $\mathbb R_{pqr}$:

\begin{equation}\label{signature_matrices}
    S_i = \begin{cases}
      P \qquad e_i^{\,2} = +1 \\
      Q \qquad e_i^{\,2} = -1 \\
      R \qquad e_i^{\,2} = 0 
    \end{cases}
\end{equation}

\noindent
Next, the basis matrices $\mathbf E_i$ corresponding to the basis vectors $\mathbf e_i$ are constructed using the tensor (Kronecker) product for matrices:

\begin{equation*}
    \mathbf E_i = \overbrace{C^i_1 \otimes ... \otimes C^i_{i-1} }^{(i-1)\,\text{times}} \otimes S_i \otimes \overbrace{I' \otimes ... \otimes I'}^{(n-i)\,\text{times}},
\end{equation*}

\noindent
where $C^j_i$ is $I$ if the corresponding $i,j$ basis elements should anti-commute, and $I'$ when they should commute.
For Clifford algebras over the reals, $C^i_j = I$.

Matrix representations for higher grade basisvectors $\mathbf e_{ij\cdots k}$ are now constructed with standard matrix multiplication: 
\begin{align}\label{basispermutation}
    \mathbf E_{ij\cdots k} = \mathbf E_i \mathbf E_j \cdots \mathbf E_k,
\end{align}
with the $2^n$ identity matrix $\mathbf  I_{2^n} = \overbrace{I \otimes ... \otimes I}^{n\, \text{times}}$ representing the scalar unit.
The resulting set of matrices 

\begin{equation}\label{basisorder}
\{ \mathbf R^i \}_{i=1}^{2^n} = \{\mathbf I_{2^n}, \mathbf E_1, \mathbf E_2, \ldots, \mathbf E_{12}, \ldots,\mathbf E_{12\cdots n}\} 
\end{equation}
\noindent 
is closed under matrix multiplication, and indeed a valid representation for the Clifford algebra \GR{pqr}. 
This set is determined up to permutations $\mathbf R^i \to  \mathbf O \mathbf R^i \mathbf O^T$, where $\mathbf O$ is a permutation matrix. Although all $\mathbf O \mathbf R^i \mathbf O^T$ are equivalent representations of \GR{pqr}, there is nonetheless a unique permutation which offers substantial computational advantages over the others.
To construct this \emph{Clifford representation} $\mathbf C^i$, we first construct a permutation matrix $\mathbf O$ with matrix elements
\begin{equation}
  \begin{split}
    \mathbf O_{ij} &= \mathbf R^i_{j0}.
  \end{split}
\end{equation}\noindent
Finally, we construct the Clifford representation $\mathbf C^i$ by conjugating each of the basis matrices $\mathbf R^i$:
\begin{equation}
  \begin{split}
    \mathbf C^i &= \mathbf O \mathbf R^i \mathbf O^T\\
  \end{split}
\end{equation}

The ordering matrix $\mathbf O$ is a permutation matrix that ensures that the only non-zero element of the first column of $\mathbf C^i$ is positive, and in the $i$-th row. The sandwich makes sure the same permutations are being applied to the rows and columns, preserving the group relations.

The resulting set of matrices $\mathbf C^i$ is the unique matrix representation of the Clifford algebra \GR{pqr} where the matrix form of a multivector has the unmodified multivector coefficients as its first column.
As a result, both the matrix-matrix and matrix-vector product represent the geometric product.

Note that this procedure allows us to easily select the metric signature of the basis vectors \cref{signature_matrices}, the permutations of the basis elements \cref{basispermutation}, and the order of all basis blades \cref{basisorder}. 

\subsection{Efficient implementation}

To demonstrate the efficiency of the Clifford representation, consider for example the matrix representation of $\mathbb R_{0,2}$, isomorphic to $SU(2)$, constructed following the procedure above for the
element $x = a + b \e{1} + c \e{2} + d \e{21}$:

\[D(\begin{bmatrix}a \\ b \\ c \\ d\end{bmatrix}) := 
\begin{bmatrix}
a &-b& -c& -d\\
b & a&  d& -c\\
c &-d&  a&  b\\
d & c& -b&  a\\
\end{bmatrix} = a \mathbf I + b \mathbf E_1 + c \mathbf E_2 + d \mathbf E_{21}.
\]
This is of course the well known real representation of the Pauli matrices \cite{bargmann,Weinberg:1995mt}, as expected. Note however, how the first column of this matrix contains, in order, all four coefficients needed to completely determine the matrix. In practice this allows storage of just the $2^n$ vector coefficients, instead of the $2^{2n}$ matrix coefficients. Additionally, this enables an efficient implementation of the geometric product between two multivectors $a,b$ represented in this $2^n$ vector space as $\vec a, \vec b$ :
\[
  \vv{ab} = D(\vec a)\vec b.
\]
Apart from the clear performance improvements, the vector form also allows easy implementation of other geometric algebra features such as grade selection, reversion, conjugation,  etc.:

\[
 {\vec {\tilde x}} := \begin{bmatrix}1&0&0&0\\0&-1&0&0\\0&0&-1&0\\0&0&0&-1\end{bmatrix} \vec {x}, \quad \vec {\overline{x}} := \mqty[\dmat[0]{1,-1,-1,1}] \vec{x},\quad \ldots.
\]
It is worth noting that $D({\vec{\tilde x}}) = D(\vec{x})^T$, while many other operations such as dualisation have no straightforward matrix equivalent.

In this form, transformations of elements by rotors are still executed by conjugation, which can be reformulated to a one sided transformation by solving
\[
D\pqty{D(\vec R)\, \vec x}{\vec{\tilde R}} = A\vec x
\]
for $A$.
In the example below we work this out for the Euclidean group \Eu{3}, recovering the familiar homogeneous representations. 

\subsection{Matrix representations of \Eu{3}}\label{sec:matrix_E3}

A multivector $x \in \GR{3,0,1}$, given by
\begin{equation}\label{eq:x_E3}
    x = s + \vec{v}_o + \vec{v}_\infty + B_o + B_\infty + \cev{t}_\infty + \cev{t}_o + p,    
\end{equation}
where
\begin{alignat*}{3}
s &= x^1,
\qquad
&p &= x^{16} \e{0123}, \\
\vec{v}_o &= x^2\e1 + x^3\e2 + x^4\e3, \qquad
&\vec{v}_\infty &= x^5\e0, \\
B_o &= x^6\e{23} + x^7\e{31} + x^8\e{12}, \qquad
&B_\infty &= x^9\e{01} +x^{10}\e{02} + x^{11}\e{03}, \\
\cev{t}_o &= x^{15}\e{123}, 
\qquad
&\cev{t}_\infty &= x^{12}\e{032} + x^{13}\e{013} + x^{14}\e{021},
\end{alignat*}
with the order and permutation of the basis elements carefully selected (it determines the similarity transformation), has a matrix representation $D(x)$:
\begin{equation*}
    \scriptsize \smqty[
\colora {1} & \colorb {2} & \colorb {3} & \colorb {4} & 0 & -\colorc {6} & -\colorc {7} & -\colorc {8} & 0 & 0 & 0 & 0 & 0 & 0 & -\colord {15} & 0\\[.3em]
\colorb {2} & \colora {1} & \colorc {8} & -\colorc {7} & 0 & -\colord {15} & \colorb {4} & -\colorb {3} & 0 & 0 & 0 & 0 & 0 & 0 & -\colorc {6} & 0\\[.3em]
\colorb {3} & -\colorc {8} & \colora {1} & \colorc {6} & 0 & -\colorb {4} & -\colord {15} & \colorb {2} & 0 & 0 & 0 & 0 & 0 & 0 & -\colorc {7} & 0\\[.3em]
\colorb {4} & \colorc {7} & -\colorc {6} & \colora {1} & 0 & \colorb {3} & -\colorb {2} & -\colord {15} & 0 & 0 & 0 & 0 & 0 & 0 & -\colorc {8} & 0\\[.3em]
\colorb {5} & \colorc {9} & \colorc {10} & \colorc {11} & \colora {1} & \colord {12} & \colord {13} & \colord {14} & -\colorb {2} & -\colorb {3} & -\colorb {4} & \colorc {6} & \colorc {7} & \colorc {8} & -\colore {16} & \colord {15}\\[.3em]
\colorc {6} & \colord {15} & -\colorb {4} & \colorb {3} & 0 & \colora {1} & \colorc {8} & -\colorc {7} & 0 & 0 & 0 & 0 & 0 & 0 & \colorb {2} & 0\\[.3em]
\colorc {7} & \colorb {4} & \colord {15} & -\colorb {2} & 0 & -\colorc {8} & \colora {1} & \colorc {6} & 0 & 0 & 0 & 0 & 0 & 0 & \colorb {3} & 0\\[.3em]
\colorc {8} & -\colorb {3} & \colorb {2} & \colord {15} & 0 & \colorc {7} & -\colorc {6} & \colora {1} & 0 & 0 & 0 & 0 & 0 & 0 & \colorb {4} & 0\\[.3em]
\colorc {9} & \colorb {5} & -\colord {14} & \colord {13} & -\colorb {2} & -\colore {16} & \colorc {11} & -\colorc {10} & \colora {1} & \colorc {8} & -\colorc {7} & -\colord {15} & \colorb {4} & -\colorb {3} & \colord {12} & -\colorc {6}\\[.3em]
\colorc {10} & \colord {14} & \colorb {5} & -\colord {12} & -\colorb {3} & -\colorc {11} & -\colore {16} & \colorc {9} & -\colorc {8} & \colora {1} & \colorc {6} & -\colorb {4} & -\colord {15} & \colorb {2} & \colord {13} & -\colorc {7}\\[.3em]
\colorc {11} & -\colord {13} & \colord {12} & \colorb {5} & -\colorb {4} & \colorc {10} & -\colorc {9} & -\colore {16} & \colorc {7} & -\colorc {6} & \colora {1} & \colorb {3} & -\colorb {2} & -\colord {15} & \colord {14} & -\colorc {8}\\[.3em]
\colord {12} & -\colore {16} & \colorc {11} & -\colorc {10} & -\colorc {6} & -\colorb {5} & \colord {14} & -\colord {13} & \colord {15} & -\colorb {4} & \colorb {3} & \colora {1} & \colorc {8} & -\colorc {7} & -\colorc {9} & \colorb {2}\\[.3em]
\colord {13} & -\colorc {11} & -\colore {16} & \colorc {9} & -\colorc {7} & -\colord {14} & -\colorb {5} & \colord {12} & \colorb {4} & \colord {15} & -\colorb {2} & -\colorc {8} & \colora {1} & \colorc {6} & -\colorc {10} & \colorb {3}\\[.3em]
\colord {14} & \colorc {10} & -\colorc {9} & -\colore {16} & -\colorc {8} & \colord {13} & -\colord {12} & -\colorb {5} & -\colorb {3} & \colorb {2} & \colord {15} & \colorc {7} & -\colorc {6} & \colora {1} & -\colorc {11} & \colorb {4}\\[.3em]
\colord {15} & \colorc {6} & \colorc {7} & \colorc {8} & 0 & \colorb {2} & \colorb {3} & \colorb {4} & 0 & 0 & 0 & 0 & 0 & 0 & \colora {1} & 0\\[.3em]
\colore {16} & -\colord {12} & -\colord {13} & -\colord {14} & -\colord {15} & \colorc {9} & \colorc {10} & \colorc {11} & \colorc {6} & \colorc {7} & \colorc {8} & \colorb {2} & \colorb {3} & \colorb {4} & \colorb {5} & \colora {1}
    ].
\end{equation*}
It is apparent from this matrix that the first column vector $\vec{x}$ lists all the coefficients, and can thus be used as a representation of $x$.

Because a multivector $x \in \GR{3,0,1}$ can be represented by $\vec{x}$,
we would like to represent the conjugation of a multivector $x$ with a normalized rotor $R \in \SE{3}$, $x \mapsto R x \widetilde{R}$, as a matrix-vector product 
\[ D\pqty{D(\vec R) \, \vec x}{\vec{\tilde R}} = A \vec{x}.\]
With this carefully chosen basis, $A$ reveals five familiar matrix representations of $\SE{3}$. 
To construct $A$, we start with the rotor $R$:
\[
R = a + b\e{23} + c\e{31} + d\e{12} + e\e{01} + f\e{02} + g\e{03} + h\e{0123},
\]
and a general multivector $x$ as defined in \cref{eq:x_E3}:
\[
x = s + \vec{v}_o + \vec{v}_\infty + B_o + B_\infty + \cev{t}_\infty + \cev{t}_o + p.
\]
By symbolically solving the linear system $A  \vec{x} = D\pqty{D(\vec{R}) \, \vec{x}} \tilde{\vec{R}}$ for $A$, we find that the matrix representation of conjugation is given by the matrix-vector product
\begin{equation}\label{eq:matrix_E3}
    A \vec{x} =  \mqty[\dmat{
    1, 
    \textcolor{RedOrange}{\mathbf{R}} & \textcolor{RedOrange}{\mathbf{0}} \\ \textcolor{RedOrange}{\mathbf{s}^\top \mathbf R} & \textcolor{RedOrange}{1}, 
    \textcolor{MidnightBlue}{\mathbf R} & \textcolor{MidnightBlue}{\mathbf{0}} \\ \textcolor{MidnightBlue}{\mathbf T \mathbf R} & \textcolor{MidnightBlue}{\mathbf R}, 
    \textcolor{Fuchsia}{\mathbf R} & \textcolor{Fuchsia}{\mathbf{t}} \\ \textcolor{Fuchsia}{\mathbf{0}} & \textcolor{Fuchsia}{1}, 
    1
    }] 
\mqty[
s\\ \vec{v}_o \\ \vec{v}_\infty \\
B_o \\ B_\infty \\ \cev{t}_\infty \\ \cev{t}_o \\ p 
],
\end{equation}
where
\begin{align*}
    \mathbf R &= \smqty[{
    a^2+b^2-c^2-d^2 & 2(ad+bc) & 2(-ac+bd) \\
    2(-ad+bc) & a^2-b^2+c^2-d^2 & 2(ab+cd) \\
    2(ac+bd) & 2(-ab+cd) & a^2-b^2-c^2+d^2 
    }], \\
    \mathbf T&=\smqty[{
    0 & t_3 & -t_2 \\
    -t_3 & 0 & t_1 \\
    t_2 & -t_1 & 0
    }],
    \, \mathbf{t} = 2\smqty[{
    cg-ae-bh-df \\ de-af-bg-ch \\ bf-ag-ce-dh
    }],
    \, \mathbf{s} = 2\smqty[{
    cg+ae+bh-df \\ de+af-bg+ch \\ bf+ag-ce+dh
    }].
\end{align*}
The matrix $A$ is block diagonal, corresponding to the fact that conjugation is a grade preserving operation. The block matrices on the diagonal are $1\times 1, 4\times 4, 6\times 6, 4\times 4$ and $1\times 1$ dimensional, and are all homogeneous matrix representations of $\SE3$ acting on invariant vectorspaces of the corresponding dimension: $s$, $\vec{v} = \vec{v}_o + \vec{v}_\infty$, $B = B_o + B_\infty$,  $\cev{t} = \cev{t}_o + \cev{t}_\infty$, and $p$ respectively. They are given by
\begin{alignat*}{3}
    D^{1}(R) &= D^{5}(R) = 1, \;
    & D^{2}(R) &= \mqty[\textcolor{RedOrange}{\mathbf{R}} & \textcolor{RedOrange}{\mathbf{0}} \\ \textcolor{RedOrange}{\mathbf{t}^\top \mathbf R} & \textcolor{RedOrange}{1}], \\
    D^{3}(R) &= \mqty[\textcolor{MidnightBlue}{\mathbf R} & \textcolor{MidnightBlue}{\mathbf{0}} \\ \textcolor{MidnightBlue}{\mathbf T \mathbf R} & \textcolor{MidnightBlue}{\mathbf R}], \;
    & D^{4}(R) &= 
    \mqty[\textcolor{Fuchsia}{\mathbf R} & \textcolor{Fuchsia}{\mathbf{t}} \\ \textcolor{Fuchsia}{\mathbf{0}} & \textcolor{Fuchsia}{1}].
\end{alignat*}
Firstly, the scalar and pseudoscalar transform under the trivial representation.
Secondly, the matrix $D^4(R)$ is the familiar $4 \times 4$ covariant homogeneous representation of the Euclidean group acting on a 4D vector space of homogeneous points:
\begin{equation*}
    \mqty[\cev{t}_\infty' \\ \cev{t}_o'] = \mqty[\textcolor{Fuchsia}{\mathbf R} & \textcolor{Fuchsia}{\mathbf{t}} \\ \textcolor{Fuchsia}{\mathbf{0}} & \textcolor{Fuchsia}{1}] \mqty[\cev{t}_\infty \\ \cev{t}_o].
\end{equation*}
Thirdly, the matrix $D^2(R)$ is the $4 \times 4$ contravariant homogeneous representation acting on a 4D vector space of homogeneous planes:
\begin{equation*}
    \mqty[\vec{v}_o' \\ \vec{v}_\infty'] = \mqty[\textcolor{RedOrange}{\mathbf{R}} & \textcolor{RedOrange}{\mathbf{0}} \\ \textcolor{RedOrange}{\mathbf{t}^\top \mathbf R} & \textcolor{RedOrange}{1}] \mqty[\vec{v}_o \\ \vec{v}_\infty].
\end{equation*}
Lastly, the matrix $D^3(R)$ is the $6 \times 6$ adjoint representation acting on a 6D vector space of Pl\"ucker line coordinates \cite{Selig2006Robotics}:
\begin{equation*}
    \mqty[B_o' \\ B_\infty'] = \mqty[\textcolor{MidnightBlue}{\mathbf R} & \textcolor{MidnightBlue}{\mathbf{0}} \\ \textcolor{MidnightBlue}{\mathbf T \mathbf R} & \textcolor{MidnightBlue}{\mathbf R}] \mqty[B_o \\ B_\infty].
\end{equation*}
These well known representations, represent the group action of $R$ on the corresponding invariant vectorspaces, which are traditionally identified as the elements of geometry.
The matrix $4 \times 4$ representations to transform points and planes already require $16$ real numbers: the same as the entire multivector $x$.
Additionally, how would we intersect a $6$ dimensional line with a $4$ dimensional plane? In the matrix formulation this is a difficult question to answer, whereas using the multivector approach the answer is simply $B \wedge \cev{t}$. The GA formulation therefore offers both conceptual clarity and computational advantages over the matrix formulation.

\section{Conclusion}

\epigraph{
\begin{tabular}{p{0.45\textwidth}|p{0.45\textwidth}}
    L'application des mêmes idées de dualité peut s'étendre à la Mécanique. En effet, l'élément primitif des corps auquel on applique d'abord les premiers principes de cette science, est, comme dans la Géométrie ancienne, le point mathématique. 
Ne sommes-nous pas autorisés à penser, maintenant, qu'en prenant le plan pour l'élément de l'étendue, et non plus le point, on sera conduit à d'autres doctrines, faisant pour ainsi dire une nouvelle science? 
    & The application of the same ideas of duality can be extended to Mechanics. Indeed, the primitive element of bodies to which the first principles of this science are applied is, as in ancient Geometry, the mathematical point. Are we not permitted to think, now, that by taking the plane for the basic element, and no longer the point, we shall be led to other doctrines, making, as it were, a new science?
\end{tabular}
}{Michel Chasles (1875, \cite{chasles1875})}
Based on the results of this paper, we feel emboldened when we answer our confr\`ere: not only is it permitted to take planes as the basis elements, it is extremely fruitful. Only when vectors are associated with hyperplanes, can we develop intuitions about transformations and elements of geometry which carry over to spaces of any metric and number of dimensions, intuitions which in turn have been crucial to the development of the invariant decomposition.

A plane based approach has been taken by geometers since time immemorial, but only in recent work by amongst others Jon Selig \cite{Selig2000CliffordAO} and Charles Gunn \cite{GunnThesis} has it been recognized that planes map elegantly onto vectors in \GR{3,0,1}, i.e. 3DPGA.
The current work demonstrates that this approach works more generally: vectors can be associated with reflections and/in hyperplanes for pseudo-euclidean groups, or inversions and/in  hyperspheres for the conformal group.

Hyperplanes intersect to form geometric elements with decreasing degrees of freedom, such as hyperlines, hyperpoints, etc. This behavior is mirrored by the outer product of vectors, which wedges vectors into bivectors, trivectors, etc., which makes their identification natural.

The composition of an even number of reflections led naturally to \Spin{p,q,r} elements, which are generated by bivector elements of the Lie algebra \spin{p,q,r}. 
Because the geometric product eliminates any identical reflections in the $2k$-reflection, it necessarily follows that the resulting rotor is always decomposable into at most $k$ mutually commuting orthogonal bireflections. 
The complementary statement in the Lie algebra \spin{p,q,r}, is that the generating bivector can be split into $k$ mutually commuting orthogonal simple bivectors. 
Additionally, the composition of $2k+1$ reflections can be decomposed into a commuting reflection and  $2k$-reflection, which can then be decomposed further into $k$ commuting bireflections.

This insight led to the invariant decomposition: a novel algorithm to either split a bivector into mutually commuting orthogonal simple bivectors (\cref{sec:invariant_decomposition}),
or to split a rotor into mutually commuting orthogonal simple rotors (\cref{sec:factor_even}). 

Because simple bivectors square to scalars, simple rotors, as the exponential of a simple bivector, follow a generalized Euler's formula. Therefore, after the decomposition has been performed, the exponential and logarithmic functions are no more complicated than those of complex analysis (\cref{sec:exponential,sec:logarithm} respectively).

The famous Mozzi-Chasles theorem is now understood to be a special case of the invariant decomposition, and it appears that Michel Chasles' suggestion that perhaps vectors should have been identified with planes, not points, does indeed offer significant advantages.

\begin{acks}
The authors would like to thank Dr. Ir. Leo Dorst for invaluable discussions about this research.
The research of M.~R. was supported by 
\grantsponsor{1}{KU Leuven}{} IF project \grantnum[]{1}{C14/16/067}.
\end{acks}

\bibliographystyle{ACM-Reference-Format}
\bibliography{biblio.bib}

\newpage

\begin{landscape}
\includepdf[landscape,fitpaper,templatesize={11.7in}{8.27in}]{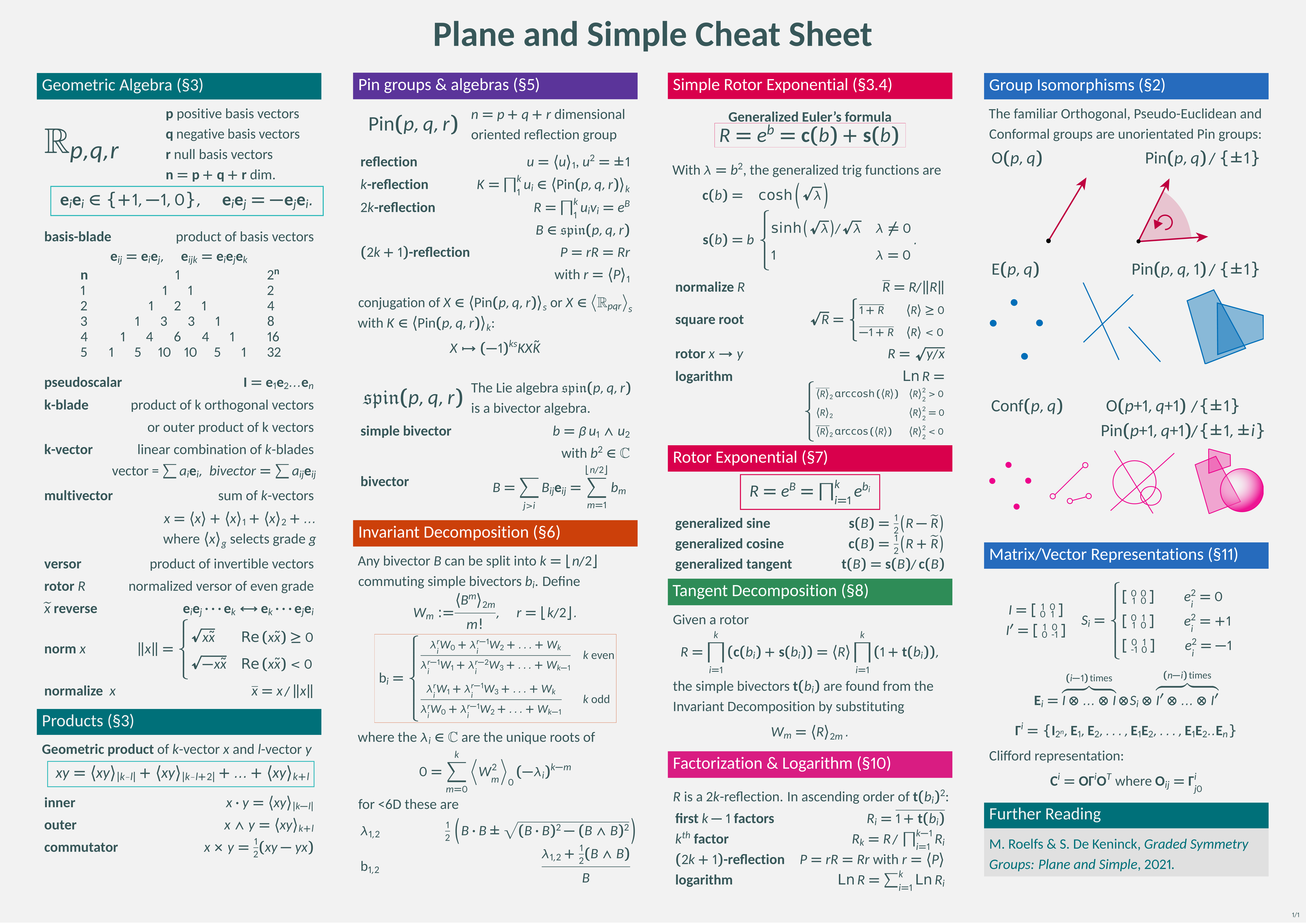}
\end{landscape}

\end{document}